\pdfoutput=1

\documentclass[aps,pra,amsmath,amssymb,twocolumn,10pt]{revtex4-2}
\usepackage{mathrsfs}
\usepackage[T1]{fontenc}
\fontencoding{T1}  
\usepackage[utf8]{inputenc}
\usepackage{enumitem}
\usepackage{tikz}

\newcommand{\vertiii}[1]{{\left\vert\kern-0.25ex\left\vert\kern-0.25ex\left\vert #1 
    \right\vert\kern-0.25ex\right\vert\kern-0.25ex\right\vert}}

\usepackage{amsfonts}
\usepackage{amsmath,amssymb,amsthm}
\usepackage{graphicx}
\usepackage{fancyhdr}
\usepackage[breakable]{tcolorbox}
\usepackage{bbm}
\usepackage{url}
\usepackage{mathtools}
\usepackage{braket}
\usepackage{upgreek }
\usepackage{dsfont}
\usepackage{collectbox}
\usepackage{braket}
\usepackage{quantikz}
\usepackage{subcaption}
\usepackage[plain]{fancyref} %references with object type, no pages
\usepackage[ruled]{algorithm2e}

\usepackage{algpseudocode}

\usepackage{hyperref}

\newtheorem{thm}{Theorem}
\newtheorem*{thm*}{Theorem}
\makeatletter
\newcommand{\setthmtag}[1]{% \settheoremtag{<tag>}
  \let\oldthethm\thethm% Store \thetheorem
  \newcommand{\thethm}{#1}% Redefine it to a fixed value
  \g@addto@macro\endthm{% At \end{theorem}, ...
    \addtocounter{thm}{-1}% ...restore theorem counter value and...
    \global\let\thethm\oldthethm}% ...restore \thetheorem
  }
\makeatother

\newtheorem{prop}[thm]{Proposition}
\newtheorem*{prop*}{Proposition}
\newtheorem{lemma}[thm]{Lemma}
\newtheorem*{lemma*}{Lemma}

\newtheorem*{cor*}{Corollary}

\newtheorem*{cj*}{Conjecture}

\newtheorem*{Def*}{Definition}

\theoremstyle{definition}

\newtheorem*{rem*}{Remark}

\def\beq{\begin{equation}}
\def\eeq{\end{equation}}
\def\bq{\begin{quote}}
\def\eq{\end{quote}}
\def\ben{\begin{enumerate}}
\def\een{\end{enumerate}}
\def\bit{\begin{itemize}}
\def\eit{\end{itemize}}

\def\r|{\right|}

\newcommand{\diag}{{\operatorname{diag}}}
\newcommand{\tr}{{\operatorname{tr}}}
\newcommand{\sign}{{\operatorname{sign}}}
\newcommand{\ketbra}[1]{|#1\rangle\langle#1|}
\newcommand{\norm}[1]{\left\|#1\right\|}

\newcommand\be{\begin{equation}}
\newcommand\ee{\end{equation}}

\begin{document}
\title{Efficient simulation of parametrized quantum circuits under non-unital noise through Pauli backpropagation}

\author{Victor Martinez}
\email[Victor Martinez ]{victor.martinez@ibm.com}
\affiliation{IBM France, Avenue de l'Europe, 92275 Bois-Colombes, France}
\affiliation{Inria, ENS Lyon, UCBL, LIP, F-69342 Lyon Cedex 07, France}
\author{Armando Angrisani}
\email[Armando Angrisani ]{armando.angrisani@epfl.ch}
\affiliation{Institute of Physics, Ecole Polytechnique F\'{e}d\'{e}rale de Lausanne (EPFL), CH-1015 Lausanne, Switzerland}
\author{Ekaterina Pankovets}
\email[Ekaterina Pankovets ]{ekaterina.pankovets@epfl.ch}
\affiliation{Institute of Physics, Ecole Polytechnique F\'{e}d\'{e}rale de Lausanne (EPFL), CH-1015 Lausanne, Switzerland}
\author{Omar Fawzi}
\email[Omar Fawzi ]{omar.fawzi@ens-lyon.fr}
\affiliation{Inria, ENS Lyon, UCBL, LIP, F-69342 Lyon Cedex 07, France}
\author{\begingroup
\hypersetup{urlcolor=navyblue}
\href{https://orcid.org/0000-0001-9699-5994}{Daniel Stilck Fran\c{c}a}
\endgroup}
\email[Daniel Stilck Fran\c ca ]{dsfranca@math.ku.dk}
\affiliation{Inria, ENS Lyon, UCBL, LIP, F-69342 Lyon Cedex 07, France}
\affiliation{Department of Mathematical Sciences, University of Copenhagen, Universitetsparken 5, 2100 Denmark}

\begin{abstract}
As quantum devices continue to grow in size but remain affected by noise, it is crucial to determine when and how they can outperform classical computers on practical tasks. A central piece in this effort is to develop the most efficient classical simulation algorithms possible. Among the most promising approaches are Pauli backpropagation algorithms, which have already demonstrated their ability to efficiently simulate certain classes of parameterized quantum circuits—a leading contender for near-term quantum advantage—under random circuit assumptions and depolarizing noise. However, their efficiency was not previously established for more realistic non-unital noise models, such as amplitude damping, that better capture noise on existing hardware. Here, we close this gap by adapting Pauli backpropagation to non-unital noise, proving that it remains efficient even under these more challenging conditions. Our proof leverages a refined combinatorial analysis to handle the complexities introduced by non-unital channels, thus strengthening Pauli backpropagation as a powerful tool for simulating near-term quantum devices.
\end{abstract}

\maketitle

\section{Introduction}

With the recent substantial growth in both the size and quality of quantum hardware, the quantum community has searched for quantum algorithms that can be implemented on near-term hardware and provide an advantage over classical algorithms on practically relevant problems~\cite{Kim2023,Yu2023,CarreraVazquez2024,Glick2024}.

In this search for efficient and practical quantum computing, parameterized quantum circuits (PQC) have gained considerable attention \cite{Cerezo2021, Bharti_2022}. These circuits play a crucial role in variational quantum algorithms (VQA) and are frequently employed in areas such as optimization, with algorithms like QAOA \cite{farhi2014} and VQE \cite{Peruzzo_2014,Cerezo_2022}, as well as in machine learning \cite{bowles2024,Glick2024}. These algorithms work by optimising the parameters, often rotation angles, of the quantum circuit to minimise or maximise an objective function, made of expectation values of observables. Initially thought to be somewhat noise resistant, strong limitations on their performance in the presence of noise were later shown \cite{Wang2021}, essentially showing that the outputs of these algorithms offer no quantum advantage at constant~\cite{StilckFrana2021} depth or moderately above that~\cite{Takagi2022} and other works established that noisy quantum computers can be simulated in various regimes. However, until recently~\cite{schuster2024,mele2024}, the vast majority of such results was restricted to noise that is \emph{unital}, i.e. maps the maximally mixed state to itself. However, the noise in current devices often does not satisfy this assumption~\cite{Chirolli2008,BlumeKohout2017}. Thus, it is imperative to understand to what extent we can simulate noisy quantum computation under non-unital noise. It is important to note that extending classical algorithms 
to such noise is far from being straightforward. Indeed, the effect of non-unital noise is known to be much more subtle than unital noise and strong qualitative gaps exist between the two. For instance, as shown in~\cite{benor2013}, without access to fresh qubits we can only implement log depth circuits under depolarizing noise, whereas we can compute for exponential times with non-unital noise.

In this work, we will show how to efficiently simulate noisy PQCs under non-unital noise using Pauli backpropagation methods. Among the various simulation approaches for VQAs, Pauli backpropagation methods are showing significant promise in practice and offer strong theoretical guarantees~\cite{fontana2023,rudolph2023, begusic2024,angrisani2024,mele2024,Shao2024,schuster2024,lerch2024,bermejo2024,rall2019,aharonov2022,gonzalezgarcía2024paulipathsimulationsnoisy}, showcasing great accuracy on systems with up to 127 qubits, even in noiseless conditions \cite{rudolph2023,Shao2024,angrisani2024, tanggara2024classically}. This approach focuses on structured circuits composed of alternating layers of parameterized rotations and Clifford operations. These circuits, which form a universal family, closely resemble the ones commonly used in PQCs. Using the circuit structure and the effect of unital noise, it has been shown that these simulation methods can recover the energy landscape of the PQC on average over the parameters in polynomial time~\cite{fontana2023,Shao2024}.

Among the Pauli backpropagation methods, LOWESA~\cite{fontana2023,rudolph2023,lerch2024,bermejo2024} has demonstrated several notable advantages. The algorithm is both memory-efficient and parallelizable. Additionally, the most resource-intensive part of the algorithm, Pauli backpropagation, only depends on the structure of the circuit and not on the values of the angles. Thus, the algorithm can be seen as a method to efficiently estimate the energy landscape as a function of the angles. Finally, the runtime of the algorithm is independent of both the quantum chip geometry and the depth of the quantum circuit.

From a technical viewpoint, the proof techniques we use require dealing with combinatorial arguments in a more subtle way than in \cite{fontana2023}, along with new Monte-Carlo techniques to handle the non-unital component of the noise. The main reason for that is that depolarizing noise uniformly dampens all possible Pauli strings, whereas non-unital noise does not. Nevertheless, our final result has a similar scaling as the ones of~\cite{fontana2023}. 

Another extension of Pauli backpropagation to random circuits with arbitrary uniform noise, possibly non-unital or dephasing, was also provided in the related paper~\cite{angrisani2025simulating}, which leverages the deterministic low-weight truncation scheme previously employed in Refs.~\cite{aharonov2022, schuster2024, gonzalezgarcía2024paulipathsimulationsnoisy}.

\section{Setup}

As in \cite{fontana2023,rudolph2023,lerch2024}, the parametrized quantum circuits we consider are made of $m$ alternating layers of single-qubit rotations around the Z-axis and Clifford operations $C_i$. The rotations are parametrized by angles $\theta=(\theta_1,\ldots,\theta_m)$ and qubit labels $(q_1, 
\ldots, q_m)$ so that the single-qubit gate rotation at layer $i$ is given by $R_z^{(q_i)}(\theta_i)=e^{-j\theta_i/2Z_{q_i}}$. The circuits can be written in the following way, and are depicted in Figure \ref{fig:circuitLOWESA}.

\begin{equation}
    U(\mathbf{\theta})=\bigg(\prod_{i=1}^m C_iR_z^{(q_i)}(\theta_i)\bigg)C_0.
\end{equation}

Note that these circuits form a universal family of quantum circuits, and that the complexity of classically simulating such quantum circuits is exponential in the number of $R_z$ rotations $m$ \cite{Bravyi_2016}. In the following discussions, we will focus only on the gate noise affecting the rotations, as it is often the case that noise on the Clifford operations can be more easily managed since they can be implemented transversally \cite{Bravyi2005}. For simplicity, the noise channel we consider is the amplitude damping channel $\mathcal{N}_{AD}$. However, we show in the supplemental material that our results hold for most single-qubit noise models considered in the literature, including the composition and probabilistic combinations of amplitude damping, dephasing and depolarizing.
Our noisy quantum circuit is thus represented by the channel

\begin{equation}\label{eq:unitaryLOWESA}
    \mathcal{U_{\theta}}=\bigg(\bigcirc_{i=1}^m\mathcal{C}_i\circ \mathcal{R}^{(q_i)}_z(\theta_i)\circ\mathcal{N}_{AD}\bigg)\circ\mathcal{C}_0,
\end{equation}
where we always denote in caligraphic font the channel given by conjugation by the corresponding unitary.

\begin{figure}
    \centering
    \includegraphics[width=0.46\textwidth]{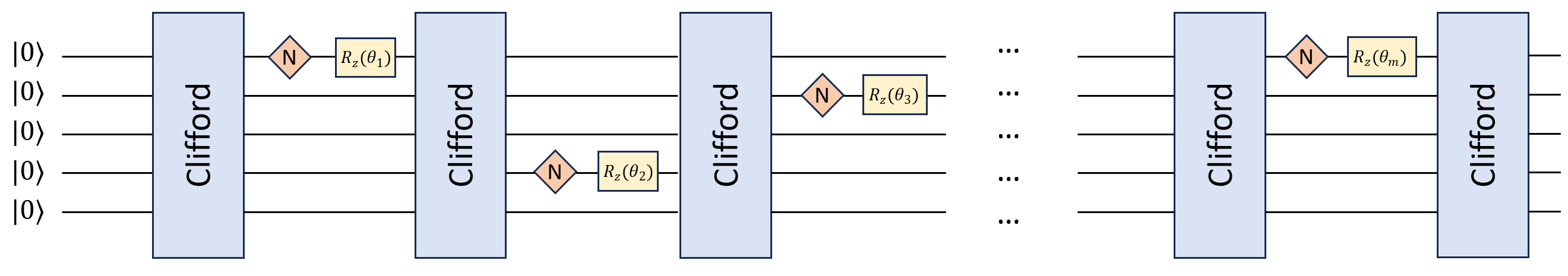}
    \caption{Representation of parametrised quantum circuits simulated by LOWESA, made of alternating layers of Clifford operations,  amplitude-damping noise and single-qubit rotations}
    \label{fig:circuitLOWESA}
\end{figure}

In most VQAs, the algorithm requires applying the quantum circuit to the initial state $\ket{0}^{\otimes n}$ and then measuring an observable~\cite{farhi2014,Peruzzo_2014,Cerezo_2022}. Our classical algorithm will estimate the expectation value of this observable. Denote by $\mathbb{P}$ the set of Pauli matrices $\mathbb{P}=\{I, X, Y, Z\}$ and by $P$ a Pauli string in $\mathbb{P}^{\otimes n}$, which will be the observable we consider. In practice, if the observable is not a Pauli string, one can always still decompose it as a sum of Pauli strings and then run the simulation procedure for each of the Pauli strings, as long as there are polynomially many of them. Our quantity of interest is the expectation value as a function of the parameters $\theta\in[0,2\pi]^m$:

\begin{equation}
    f(\theta)=\tr(\mathcal{U_{\theta}}(\ketbra{0})P).
\end{equation}

We will approximate this expectation value by $\Tilde{f}(\theta)$ that can be computed in polynomial time. The approximation error we consider is the $L^2$-norm over the parameter space $\Theta=[0,2\pi]^m$:

\begin{equation}
    \begin{aligned}
    \Delta(f,\Tilde{f})&=||f-\Tilde{f}||_2\\
   &=\bigg(\frac{1}{|\Theta|}\int_{\Theta}|f(\theta)-\Tilde{f}(\theta)|^2d\theta\bigg)^{1/2}.
\end{aligned}
\end{equation}

Note that convergence in this metric does not imply that our approximation scheme succeeds for a given set of angles.

When backpropagating observables in the Heisenberg picture, it is often useful to use the Pauli Transfer Matrix (PTM) formalism \cite{Chow2012}, which we briefly introduce here. In the PTM, operators and states are represented in the Pauli basis, such that the PTM representation of a quantum state $\rho$ is the real valued vector $\ket{\rho}\rangle\in\mathbb{R}^{2^n}$ where $[\ket{\rho}\rangle]_i=\tr(\rho P_i)$ for $P_i\in\mathbb{P}^{\otimes n}$. Similarly, a quantum channel on $n$-qubits $\mathcal{E}$ is represented by the matrix $\mathbf{E}$ with $\mathbf{E}_{ij}=2^{-n}\tr(P_i\mathcal{E}(P_j))$. Using this formalism, it is clear that the expectation value of a Pauli string can be written as $\tr(P\mathcal{E}(\rho))=\langle\bra{P}\mathbf{E}\ket{\rho}\rangle$. Computing expectation values in the Heisenberg picture, where quantum channels act on measurement operators instead of quantum states, can be done using this formalism.

Indeed, the corresponding adjoint operation simply requires taking the transpose of the expectation value,
$\langle\bra{P}\mathbf{E}\ket{\rho}\rangle=\langle\bra{\rho}\mathbf{E}^T\ket{P}\rangle$. If $\mathcal{E}$ is made of multiple 2-qubit Clifford unitaries, this gives an efficient way of computing expectation values. Indeed, in the PTM, each Clifford unitary takes the Pauli operator $P$ to another Pauli operator $P'$, up to a phase. For 2-qubit unitaries keeping track of this change takes at most $\mathcal{O}(n^2)$ \cite{Aaronson_2004}. As long as there are $\textrm{poly}(n)$ 2-qubit Clifford unitaries, this gives us an efficient way of computing expectation values of Clifford circuits in the Heisenberg picture. Handling the rotations, which are not Clifford gates, requires more work and the use of the structure of the noise channel.

\section{Main result}

To deal with the rotations, we start by noting that the adjoint of the amplitude damping channel with damping parameter $\gamma>0$ can easily be defined in the PTM since,

\begin{equation}
\begin{aligned}
    \mathcal{N}^\dagger_{AD}(X)&=\sqrt{1-\gamma}X, \quad \mathcal{N}^\dagger_{AD}(Y)=\sqrt{1-\gamma}Y,\\
    \quad \mathcal{N}^\dagger_{AD}(Z)&=(1-\gamma)Z+\gamma I,  \quad \mathcal{N}_{AD}^\dagger(I)=I
\end{aligned}.
\end{equation}

Let $\mathbf{R_z}(\theta)$ be the PTM representation of the adjoint of $\mathcal{R}_z(\theta)$ and $\mathbf{N_{AD}}$ be the PTM representation of the adjoint of the noise channel. The noisy rotations can be explictly computed,

\begin{equation}
\begin{aligned}
 \hspace{-0.1cm}\mathbf{R_Z}\cdot&\mathbf{N_{AD}}=\begin{bmatrix}
1 & 0 & 0 & \gamma \\
0 & \sqrt{1-\gamma}\cos{\theta} & -\sqrt{1-\gamma}\sin{\theta} &0 \\ 
0 & \sqrt{1-\gamma}\sin{\theta} & \sqrt{1-\gamma}\cos{\theta} & 0 \\
0 & 0 & 0 & 1-\gamma
\end{bmatrix}\\
&= \mathbf{D_0} + \mathbf{D_{0_Z}} + \mathbf{D_{0_I}} + \mathbf{D_{1}} \cos{\theta} + \mathbf{D_{-1}} \sin{\theta}
\end{aligned}
\end{equation}

\noindent where we have defined the quantum processes $\mathbf{D_0}=\ketbra{I}$, $\mathbf{D_{0_Z}}=(1-\gamma)\ketbra{Z}$, $\mathbf{D_{0_I}}=\gamma\ket{I}\bra{Z}$, $\mathbf{D_1}=\sqrt{1-\gamma}(\ketbra{X}+\ketbra{Y})$ and $\mathbf{D_{-1}}=\sqrt{1-\gamma}(\ket{Y}\bra{X}-\ket{X}\bra{Y})$, where $\ket{I} = [1 \ 0 \ 0 \ 0]^T, \ket{X}  = [0 \ 1 \ 0 \ 0]^T \ldots$

Therefore, when backpropagating through the circuit in the Heisenberg picture, the noisy rotation on qubit $q_i$ acts in the following way. If $P_{q_i}=I$, nothing happens, and $P_{q_i}$ remains unchanged through $\mathbf{D_0}$. If $P_{q_i}=X$, then $P_{q_i}$ is mapped to $P'_{q_i}=\sqrt{1-\gamma}\cos{\theta}X$ through process $\mathbf{D_1}$ or to $P'_{q_i}=\sqrt{1-\gamma}\sin{\theta}Y$ through process $\mathbf{D_{-1}}$. Similarly, if $P_{q_i}=Y$, then it is mapped to $P'_{q_i}=\sqrt{1-\gamma}\cos{\theta}Y$ through process $\mathbf{D_1}$ or to $P'_{q_i}=-\sqrt{1-\gamma}\sin{\theta}X$ through process $\mathbf{D_{-1}}$. Finally, if $P_{q_i}=Z$, then it is either mapped to $P'_{q_i}=\gamma I$ through $\mathbf{D_{0_I}}$ or to $P'_{q_i}=(1-\gamma)Z$ through $\mathbf{D_{0_Z}}$.

Except for $I$, we see that when going through the noisy rotation, every Pauli is split into two downscaled Paulis. This gives us a way of computing the expectation value of our noisy circuit. Start with our observable, which is a Pauli string $P$, then backpropagate it through each layer backwards, alternating layers of Clifford unitaries and single-qubit rotations. The Clifford unitaries update the Pauli string into another Pauli string - up to a phase factor - which can be computed in time $\mathcal{O}(n^2)$. The single-qubit rotation $R_z^{(q_i)}$ takes the Pauli string $P$ into potentially 2 paths, corresponding to different Pauli strings differing on qubit $q_i$. We then keep track of the new paths created and reiterate this procedure. This creates a tree structure $T$, depicted in Figure \ref{fig:PathsSplitting},  which may have up to $2^m$ branches, where $m$ is our number of parameterized rotations.

\begin{figure}
    \includegraphics[width=0.37\textwidth]{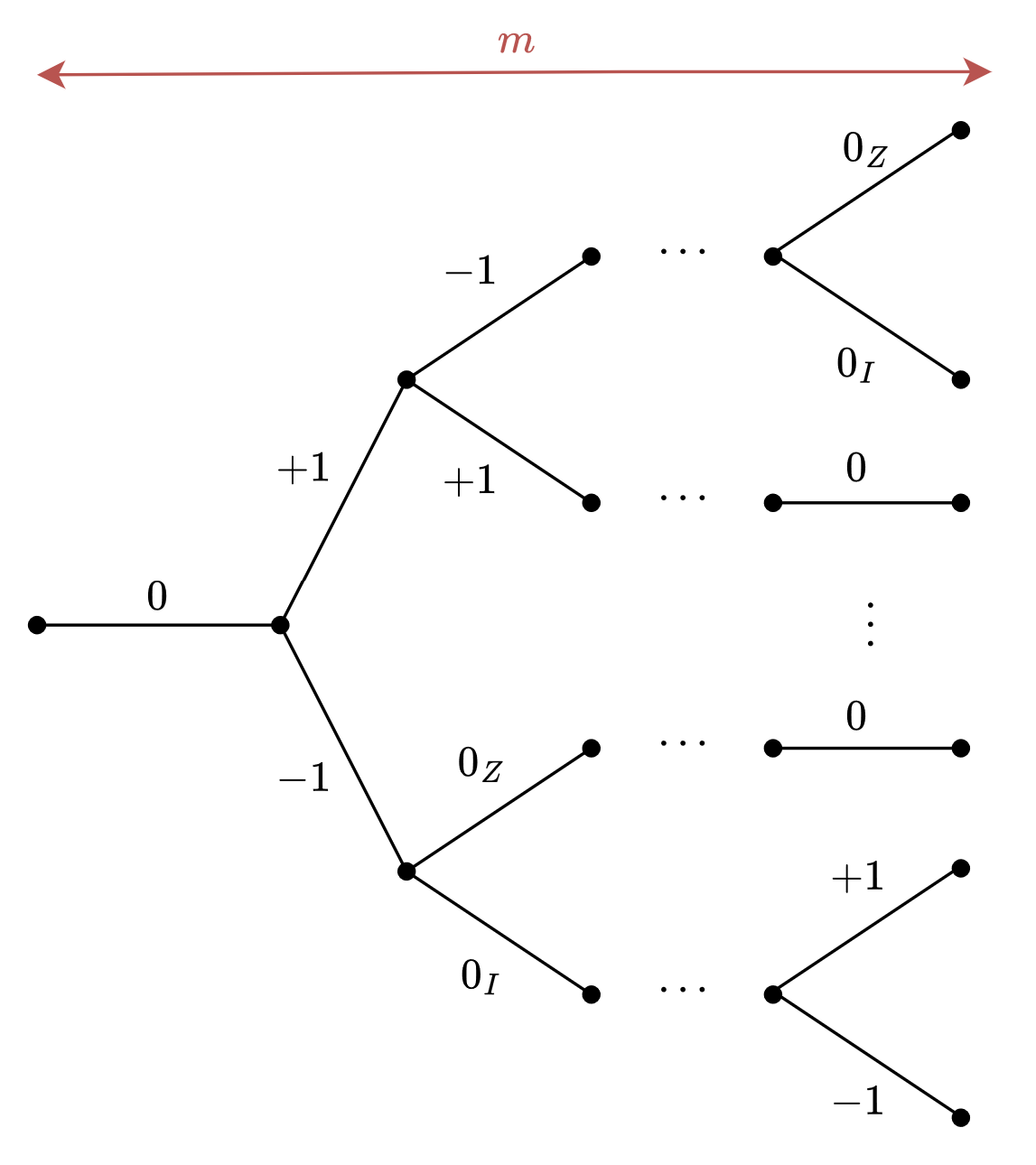}
    \caption{Tree corresponding to the backpropagation of a single  Pauli observable through the circuit. Each node corresponds to a single Pauli string backpropagated through the quantum circuits through the quantum processes $\mathcal{D}_\omega$.}
    \label{fig:PathsSplitting}
\end{figure}

Similarly to the unital case \cite{fontana2023}, formally decomposing our circuit using this tree structure requires us to keep track at each layer of both the rotation angles and which process was encountered for each path. This is done by storing, for each path (i.e. each branch of the tree), the processes that were encountered during the backpropagation in a vector $\omega \in \{0,0_Z,0_I,\pm 1\}^m$ and by storing separately the rotation angles. These two quantities are then combined using the convenient trigonometric monomials $\phi_{\omega_i}(\theta_i)$, with $\phi_0(\theta_i)=\phi_{0_Z}(\theta_i)=\phi_{0_I}(\theta_i)=1$, $\phi_{1}(\theta_i)=\cos{\theta_i}$, and $\phi_{-1}(\theta_i)=\sin{\theta_i}$. By defining the quantity $\Phi_\omega(\theta)=\prod_{i=1}^m \phi_{\omega_i}(\theta_i)$, our quantum circuit can be decomposed using the quantum processes $\mathcal{D}_\omega$, which we have previously defined in the PTM:

\begin{equation}
    \mathcal{U}_{\theta}=\sum_{\omega \in \{0,0_z,0_I,1,-1\}^m} \Phi_\omega(\theta)\mathcal{D}_\omega.
\end{equation}

\noindent Where each process $\mathcal{D}_\omega$ is given by $\mathcal{D}_\omega=(\bigcirc_{i=1}\mathcal{C}_i\circ \mathcal{D}_{\omega_i})\circ \mathcal{C}_0$. In this setting, our parametrized expectation value $f$ can then be written as:

\begin{equation}
f(\theta)=\tr(\mathcal{U_{\theta}}[\ketbra{0}]P)=\sum_\omega d_\omega\Phi_\omega(\theta)
\end{equation}

\noindent where the Fourier coefficients $d_\omega$ are given by,

\begin{equation}
    d_\omega=\tr(\mathcal{D_{\omega}}[\ketbra{0}]P)=\langle \bra{P}\mathbf{D_\omega}\ket{0}\rangle.
\end{equation}

Recall that when backpropagating a Pauli string $P$, the rotation at each layer acts only on a single qubit $q_i$. Depending on the Pauli operator at qubit $q_i$, at most two processes are valid: $\mathbf{D_0}$ for $I$, $\mathbf{D_{0_Z}}$ and $\mathbf{D_{0_I}}$ for $Z$, and $\mathbf{D_{\pm 1}}$ for $X$ and $Y$. 

To approximate this expectation value, we propose two different methods, each with its own caveats and advantages. The first method, Algorithm \ref{alg:Algorithm}, considers only the paths in the tree $T$ that have split fewer than $\ell$ times, where $\ell$ is a chosen cut-off parameter. Indeed, each split of the paths dampens the corresponding Pauli string, such that discarding paths that have split sufficiently could yield a good approximation.  Let $\#\omega$ denote the number of times a given path $\omega$ has split. Algorithm \ref{alg:Algorithm} approximates the function $f$ by the quantity:

\begin{equation}
    \Tilde{f}(\theta)=\sum_{\substack{\omega \in T \\ \#\omega\leq\ell}} d_\omega\Phi_\omega(\theta)
\end{equation}

\noindent which is computed in polynomial time for $\ell=\mathcal{O}(1)$. However, unlike in the case of unital noise, it is uncertain whether each split effectively dampens the expectation value of the Pauli observable. Specifically, when a split occurs due to encountering a Pauli $Z$ at qubit $q_i$, it results in two new branches through processes $\mathbf{D_{0_I}}$ and $\mathbf{D_{0_Z}}$. In the case of amplitude damping alone, this split may leave the approximation error unchanged. Therefore, it is essential to distinguish between splits occurring through processes $\mathbf{D_{\pm 1}}$, which always dampen the approximation error, and those occurring through $\mathbf{D_{0_Z}}$ and $\mathbf{D_{0_I}}$. This is done by noting which splits occur on discarded paths when doing the backpropagation.

The second approach, Algorithm \ref{alg:AlgorithmMC}, gets rid of this issue entirely by randomly sampling only one of the processes $\mathbf{D_{0_Z}}$/$\mathbf{D_{0_I}}$ in the tree. Any subtree $T_k$ obtained by this sampling procedure only contains splits that effectively dampen the expectation value of the observable. Each of the resulting trees are then truncated in the same manner after enough splitting has occurred. By repeating the randomized sampling process for $K$ independent trees $T_k$, our quantity of interest is approximated by the empirical average

\begin{equation}
        \Tilde{\hat{f}}(\theta)=\frac{1}{K}\sum_{k=1}^K\sum_{\substack{\omega\in T_k\\\#\omega\leq\ell}}\Phi_\omega(\theta)d_\omega
\end{equation}

We present the following theorem for the performance of both algorithms,

\begin{thm}\label{th:mainTheorem}
    Let $\mathcal{U}_\theta$ be a noisy quantum circuit as defined in Equation \ref{eq:unitaryLOWESA} with amplitude damping parameter $\gamma>0$ and $P$ a Pauli observable with expectation value $f(\theta)=\tr(\mathcal{U_{\theta}}[\ketbra{0}]P)$. For a given cut-off parameter $\ell$, Algorithm \ref{alg:Algorithm} outputs in time $\mathcal{O}(n^2m2^\ell)$ an approximation $\Tilde{f}$ and a lower bound $r$ on the minimum number of times the processes $\mathbf{D_{\pm 1}}$ have been encountered on every discarded branch during the backpropagation, such that,
    
    \begin{equation}\label{eq:certificateBound}
        \Delta(f,\Tilde{f})\leq  (1-\gamma)^{r/2}
    \end{equation}

     Given the same cut-off parameter $\ell$, and a sampling overhead $K$, Algorithm \ref{alg:AlgorithmMC} outputs in time $\mathcal{O}(Kn^2m2^\ell)$ an approximation $\hat{\Tilde{f}}$ that satisfies with probability at least $1-\delta$,
    \begin{equation}\label{eq:MCbound}
        \Delta(f,\hat{\Tilde{f}})\leq  (1-\gamma)^{(\ell+1)/2}+ \sqrt{\frac{2\log{(\delta^{-1}/2})}{K}}
    \end{equation}
    where the probability comes from the random sampling of the trees.
\end{thm}

In particular, this means that $\textit{any}$ quantum circuit of our universal family can be efficiently simulated under non-unital noise in the $L^2$-error metric. This result greatly extends what was previously known for simulation under non-unital noise, eliminating many of the assumptions on random circuit/noise introduced in prior works \cite{mele2024,schuster2024}.  Furthermore, Equation \ref{eq:certificateBound} (and \ref{eq:randomBound} in Appendix \ref{sec:ProofConvergence}) show that for most quantum circuits, celebrated simulation methods \cite{fontana2023,rudolph2023,lerch2024} are expected to work also for non-unital noise. Indeed, although the formal guarantees we give in the result above are generally better for Alg. \ref{alg:AlgorithmMC}, since it may be that $r=0$ for some circuits and observables, Alg. \ref{alg:Algorithm} is closer to what is done in practice and is deterministic. We show with numerical evidence that in practice, one expects $r$ to scale linearly with $\ell$ (see Figure \ref{fig:expected_r} in Appendix \ref{sec:Algorithm}). Equation $\ref{eq:certificateBound}$ also gives a certificate on how good Algorithm \ref{alg:Algorithm} performs for a given circuit and observable $P$, based on the computed value of $r$.
We give here a sketch of the proof and refer to Appendices \ref{sec:ProofConvergence} and \ref{sec:MonteCarlo} for more detailed ones. The complexity of the algorithms is straightforward to derive. Computing the effect of a Clifford unitary takes $\mathcal{O}(n^2)$, and this needs to be done for all paths at every layer. Since there are at most $2^\ell$ paths and $m$ layers, the total complexity is at most $\mathcal{O}(n^2m2^\ell)$. For Algorithm \ref{alg:AlgorithmMC}, the Monte-Carlo simulation requires doing this process $K$ times, such that the runtime is $\mathcal{O}(Kn^2m2^\ell)$. Let's now focus on the approximation error made by the algorithm. The $L^2$-error made by truncating a tree can be rewritten (see Lemma \ref{eq:lemmaProof}):

\begin{equation}
    \begin{aligned}                 \Delta^2(\Tilde{f},f)&=\frac{1}{|\Theta|}\int_\Theta |\Tilde{f}(\theta)-f(\theta)|^2 d\theta\\
    &=\frac{1}{|\Theta|}\int_{\Theta}\big |\sum_{\substack{\omega\in T,\\\#\omega>\ell}}\Phi_\omega(\theta)d_\omega \big |^2d\theta\\
    &=\sum_{\substack{(\omega,\omega')\in T^2, \\h(\omega)=h(\omega'), \\ \#\omega>\ell, \#\omega'>\ell}}2^{-|h(\omega)|}d_\omega \Bar{d_{\omega'}}.
\end{aligned}
\end{equation}

\noindent The function $h$ transforms a vector $\omega$ into a vector $h(\omega)$ by assigning the same value, $0$, to $0$, $0_I$, and $0_Z$, while leaving the rest unchanged. The condition $h(\omega) = h(\omega')$ ensures that only few terms remain in the error sum, and results from the averaging over the rotation angles. The magnitude $|h(\omega)|$ counts the number of non-zero elements in $\omega$, \textit{i.e.} the number of times processes $\mathbf{D_{\pm 1}}$ have been encountered during the backpropagation on path $\omega$. The error depends on all pairs of paths $(\omega, \omega')$ that satisfy certain constraints. For each of these pairs, it is possible to factor out the noise introduced by the processes from $d_\omega$. By defining $Q(\omega) = \prod_{i=1}^m Q(\omega_i)$ as the noise accumulated through the splitting, where $Q(0_Z) = 1 - \gamma$, $Q(0_I) = \gamma$, $Q(+1) = Q(-1) = \sqrt{1 - \gamma}$, and $Q(0) = 1$, we can rewrite $d_\omega$ as $d_\omega = Q(\omega) d_\omega^0$. By definition of $d_\omega^0$, it is clear that for all $\omega$, $|d_\omega^0| \leq 1$. This allows us to derive a simple bound on the error:

\begin{equation}
\Delta^2(\Tilde{f},f)\leq\sum_{\substack{(\omega,\omega')\in T^2, \\h(\omega)=h(\omega'), \\ \#\omega>\ell, \#\omega'>\ell}}2^{-|h(\omega)|}Q(\omega)Q(\omega')
\end{equation}

The error can then be decomposed in two sums, one where $\omega=\omega'$, and one where $\omega \neq \omega'$ but $h(\omega)=h(\omega')$. We show that in both cases, splits occuring through the $\mathbf{D_{\pm 1}}$ processes dampen the approximation error by $(1-\gamma)$ and splits occuring through the processes $\mathbf{D_{0_I}}$ and $\mathbf{D_{0_Z}}$ leave the approximation error unchanged. To see this, take a branch of the tree which has had $j$ splits (denote the branch length at this point by $i$), and go through the tree until one more split occurs on this branch. We show that the two resulting branches' contribution to the error term is less than the contribution of the original branch (by computing $Q(\omega)Q(\omega')$ and $2^{-|h(\omega)|}$) in both the cases introduced before, illustrated in Figure \ref{fig:Xsplit} and \ref{fig:XYsplit}. When a pair of paths $(\omega,\omega')$ has a nonzero contribution to the error term, \textit{i.e.} when $h(\omega)=h(\omega')$, we say that these paths "interact". 

\begin{itemize}
    \item \underline{$\omega_{1,\ldots,i}=\omega_{1,\ldots,i}'$:} If the split occurs through processes $\mathbf{D_{\pm 1}}$, then the two resulting branches can only interact with themselves in order to satisfy $h(\omega)=h(\omega')$, such that the damping introduced is $2\times(1-\gamma)2^{-1}=1-\gamma$. If the split comes from processes $\mathbf{D_{0_I}}$ and $\mathbf{D_{0_Z}}$, then there will be four terms in the sum as all pairs of paths will satisfy $h(\omega)=h(\omega')$. The approximation error remains unchanged since $\gamma^2+2\gamma(1-\gamma)+(1-\gamma)^2=1$.

    \begin{figure}     
    \centering
    \includegraphics[width=0.41\textwidth]{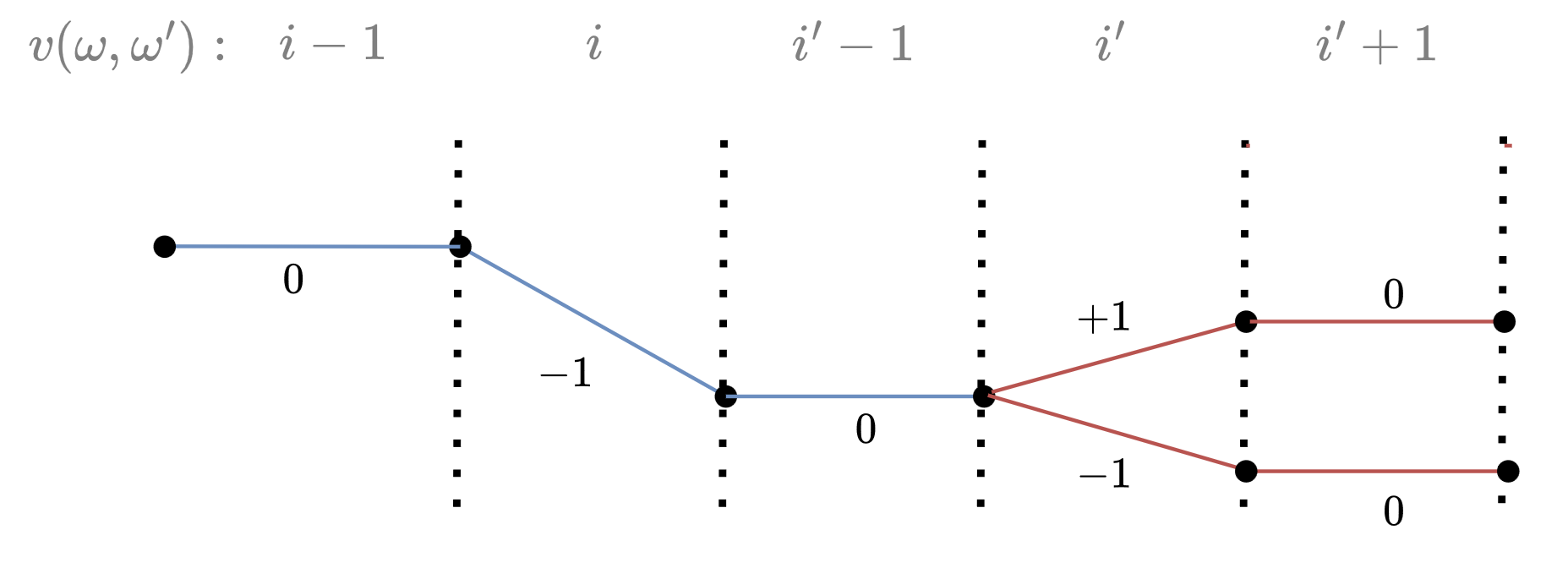}
    \caption{Branch of a tree resulting from the backpropagation of a Pauli observable, which splits through the process $\mathbf{D_{\pm1}}$.}
    \label{fig:Xsplit}
    \end{figure}

    \item \underline{$\omega_{1,\ldots,i}\neq\omega_{1,\ldots,i}'$:} The cross interactions between different branches are more subtle, as a different split (or no split) may occur on one of the branches. First, note that if one of the branches does not split through process $\mathbf{D_0}$, then to satisfy $h(\omega)=h(\omega')$, the only valid split for the other branch is through processes $\mathbf{D_{0_Z}}$ and $\mathbf{D_{0_I}}$. In this case, the two new branches will interact with the single branch and the error is unchanged since $1\times \gamma+1\times (1-\gamma)=1$. If both branches split, then to satisfy $h(\omega)=h(\omega')$, they either both split through processes $\mathbf{D_{\pm 1}}$ or both split through processes $\mathbf{D_{0_Z}}$ and $\mathbf{D_{0_I}}$. The first case leads to two interactions ($-1$ with $-1$ and $+1$ with $+1$), introducing the same damping as before $1-\gamma$. Similarly, the second case leaves the error unchanged since $\gamma^2+2\gamma(1-\gamma)+(1-\gamma)^2=1$.

\begin{figure}   
    \centering
    \includegraphics[width=0.41\textwidth]{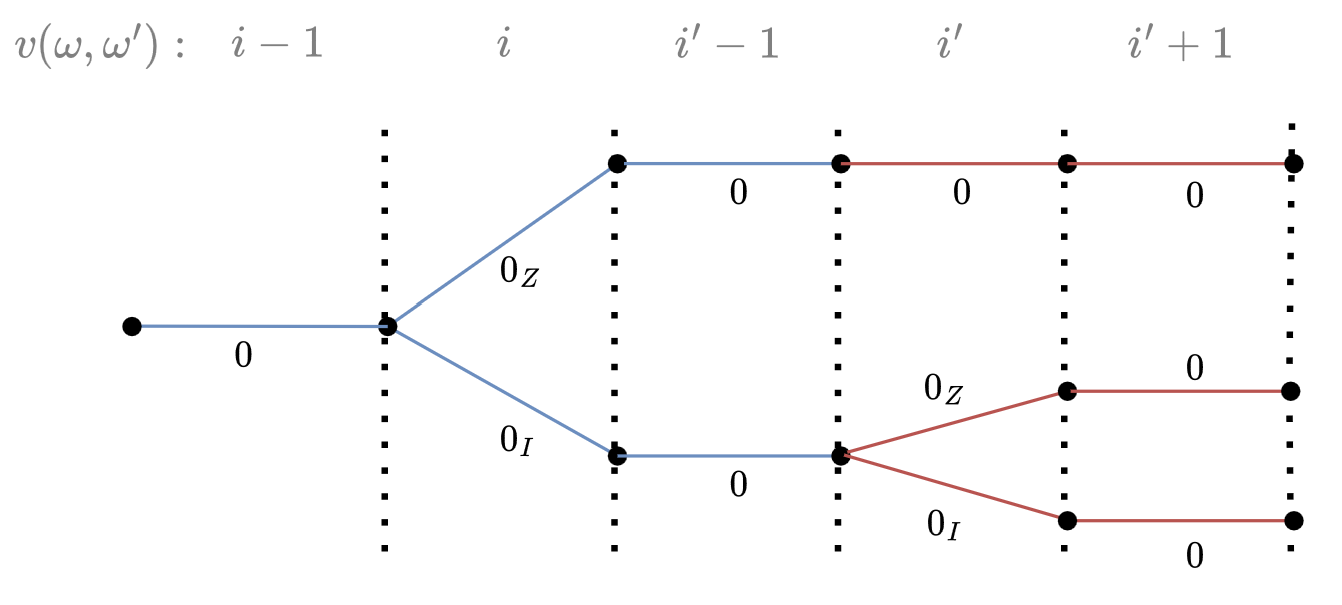}
    \includegraphics[width=0.39\textwidth]{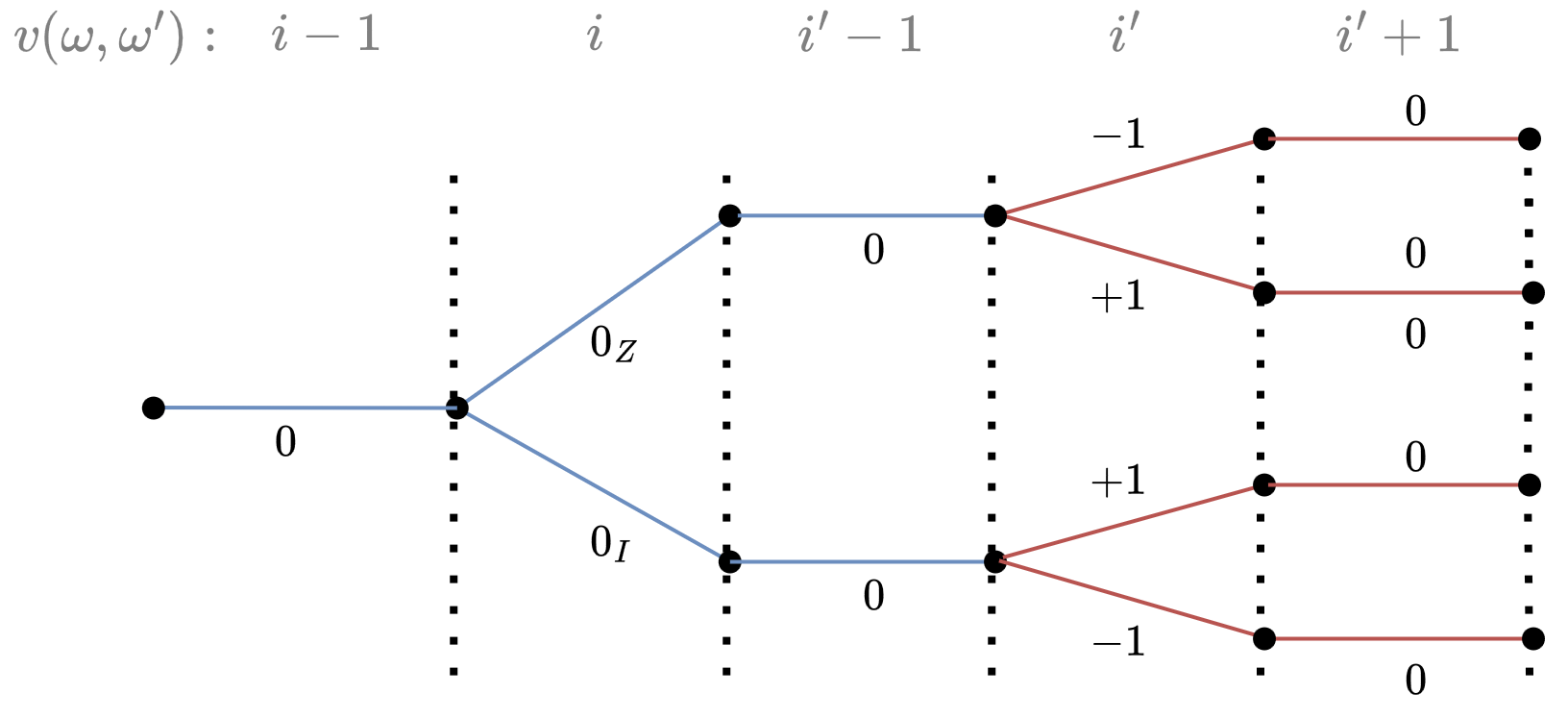}
    \caption{Two different branches of a tree resulting from the backpropogation of a Pauli observable that satisfy $h(\omega_{1,\ldots,i})=h(\omega_{1,\ldots,i}')$. To satisfy $h(\omega)=h(\omega')$ on the entire branches, only few pairs of processes are possible.}
    \label{fig:XYsplit}
    \end{figure}
\end{itemize}

Everytime a split occurs through processes $\mathbf{D_{\pm 1}}$, its total contribution is damped by a factor $1-\gamma$. Similarly, when a split occurs through processes $\mathbf{D_{0_Z}}$ or $\mathbf{D_{0_I}}$, the approximation error is unchanged. Counting for each discarded branch the number of times that processes $\mathbf{D_{\pm 1}}$ have been encountered allows us to retrieve Equation \ref{eq:certificateBound}. The convergence of the Monte-Carlo approach, Equation \ref{eq:MCbound}, can be derived similarly. The error made is decomposed in an error term arising from the Monte-Carlo approximation which scales in $1/\sqrt{K}$, our estimator being unbiased, and a truncation error (see Appendix \ref{sec:MonteCarlo}). The truncation error for each of the sampled trees is easily bounded by $(1-\gamma)^{(\ell+1)/2}$ using the same analysis and the fact that all splits in those trees occur through processes $\mathbf{D_{\pm1}}$.

\section{Conclusion}

In this work, we have shown that Pauli backpropagation methods allow us to efficiently compute the expectation value of Pauli observables under non-unital noise. This simulation method creates a surrogate of the energy landscape by only backpropagating the observable once in the Heisenberg picture. Expectation values for given angles can then be quickly obtained by just computing the trigonometric monomials. This simulation method is parallelizable and memory efficient, and does not require any assumptions on the locality of the observable, on the geometry of the underlying quantum hardware or on the depth of the quantum circuit at hand. This algorithm is particularly adapted to PQCs, as the structure strongly resemble the one commonly used for such circuits. Thus, our work establishes that Pauli backpropagation is an effective and efficient method to simulate noisy quantum devices under realistic hardware noise assumptions.

Furthermore, even though results like the quantum refrigerator~\cite{benor2013} show that the hardware noise can have significant impact on its computational power, our result clearly indicate that this difference only arises for highly structured circuits, like the ones used in fault-tolerant constructions. Finally, it would be interesting to show similar results to simulate noisy quantum simulators.

\medskip

\paragraph*{Acknowledgements:} The authors thank Manuel S. Rudolph, Zoë Holmes and Antonio Anna Mele for useuful feedback and discussion. This work benefited from a government grant managed by the Agence Nationale de la Recherche under the Plan France 2030 with the reference ANR-22-PETQ-0007. AA acknowledges support from the Sandoz Family Foundation-Monique de Meuron program for Academic Promotion.

\bibliographystyle{apsrev4-2}
\bibliography{references}

\appendix
\pagestyle{fancy}
\fancyhf{} % clear all header and footer fields
\fancyfoot[C]{\thepage} % except the center
\renewcommand{\headrulewidth}{0pt} % and the line
\renewcommand{\footrulewidth}{0pt}
\widetext

\newpage
\section{Average Pauli backpropagation algorithm}\label{sec:Algorithm}
In this section we present the pseudocode for our deterministic Pauli backpropagation algorithm in Alg.~\ref{alg:Algorithm} and empirically evaluate how fast it converges.
\begin{algorithm}
\caption{[LOWESA-AD] Simulation of Pauli string observables under non-unital noise}\label{alg:Algorithm}
\textbf{Input:} Quantum circuit made of $m$ alternating layers of Clifford unitaries and $R_z$ rotations, affected by non-unital noise, and defined by process modes $\{\mathcal{D}_\omega\}$; measurement Pauli operator $P$; cut-off parameter $\ell$.\\
\textbf{Output:} $\Tilde{f}(\theta)$, an approximation of $f(\theta)$ and $r\leq\min_{\omega /\#\omega\geq\ell} |h(\omega)|$.

\textbf{Procedure} \emph{LOWESA-AD}: Given a quantum circuit with rotation parameters $\theta$ and Pauli observable $P$, outputs an approximation of the expectation value of $P$, by first creating a surrogate of the energy landscape using \emph{Low-weight coefficients} and then adjusting the angles.

\begin{algorithmic}[1]
\State $\Tilde{f}(\theta) \leftarrow 0$
\State Run \emph{Low-weight coefficients} to calculate $d_\omega = \langle\bra{0}\mathbf{D_\omega}^T\ket{P}\rangle$ for all $\omega$ such that $\#\omega \leq \ell$.
\State \textbf{for all} $\omega$ such that $\#\omega \leq \ell$ with non-zero $d_\omega$ \textbf{do}
\State \hspace{0.5cm} $\Tilde{f}(\theta) \leftarrow \Tilde{f}(\theta) + d_\omega \Phi_\omega (\theta)$
\State \textbf{end for}
\State Return $\Tilde{f}(\theta)$ and $r=\min_{\omega}r_\omega\leq\min_{\omega /\#\omega\geq \ell} |h(\omega)|$
\end{algorithmic}
\textbf{end Procedure}\\

\textbf{Procedure} \emph{Low-weight coefficients}: Given a quantum circuit and Pauli observable $P$, outputs the backpropagated tree coefficients $d_\omega$ for all $\omega$ such that $\#\omega \leq \ell$ in time $\mathcal{O}(n^2m2^\ell)$. The tree is constructed iteratively as follow: at each layer of the quantum circuit, the tree $T_i$ is extended at each branch $\omega$ based on the Clifford layer $C_{m-i+1}$ and the Pauli operator on the rotation qubit $q_{m-i+1}$. Note that the initial tree $T_0$ is made of a single node corresponding to the Pauli observable $P$. The branches of the tree that have split more than $\ell$ times are discarded. 
\begin{algorithmic}[1]
\State \textbf{for} $i = 1$ \textbf{to} $m$ \textbf{do}
\State \hspace{0.5cm} \textbf{for} $\omega$ \textbf{in} $T_{i-1}$ \textbf{do}
\State \hspace{1cm} $P_\omega \leftarrow C_{m-i+1}^\dagger P_\omega C_{m-i+1}$, up to a phase $\phi_{\omega_i} \in \{ \pm1 \}$ which is stored.
\State \hspace{1cm} \textbf{if} $P_\omega^{q_{m-i+1}} = I$
\State \hspace{1.5cm} Add branch $\omega$ in $T_{i}$ with $\omega_i = 0$
\State \hspace{1.5cm} $d_{\omega_i}, P_\omega \leftarrow \mathcal{D}_{\omega_i}^\dagger (P_\omega)$
\State \hspace{1cm} \textbf{else if} $P_\omega^{q_{m-i+1}} = Z$ and $\#\omega < \ell$ \textbf{then}
\State \hspace{1.5cm} Split into two branches $\omega$ and $\omega'$ in $T_{i}$ with $\omega_i= 0_I$ and $\omega'_i = 0_Z$
\State \hspace{1.5cm} $d_{\omega_i}, P_\omega \leftarrow \mathcal{D}_{\omega_i}^\dagger (P_\omega)$
\State \hspace{1.5cm} $d_{\omega'_i}, P_{\omega'} \leftarrow \mathcal{D}_{\omega'_i}^\dagger (P_\omega)$
\State \hspace{1.5cm} $\#\omega,\#\omega' \leftarrow +1$
\State \hspace{1.cm} \textbf{else if} $P_\omega^{q_{m-i+1}} \in \{X,Y\}$ and $\#\omega < \ell$ \textbf{then}
\State \hspace{1.5cm} Split into two branches $\omega$ and $\omega'$ in $T_{i}$ with $\omega_i= +1$ and $\omega'_i = -1$
\State \hspace{1.5cm} $d_{\omega_i}, P_\omega \leftarrow \mathcal{D}_{\omega_i}^\dagger (P_\omega)$
\State \hspace{1.5cm} $d_{\omega'_i}, P_{\omega'} \leftarrow \mathcal{D}_{\omega'_i}^\dagger (P_\omega)$
\State \hspace{1.5cm} $\#\omega,\#\omega' \leftarrow +1$
\State \hspace{1cm} \textbf{else if} $P_\omega^{q_{m-i+1}} \in \{X,Y, Z\}$ and $\#\omega =\ell$ \textbf{then}
\State \hspace{1.5cm} Discard path $\omega$
\State \hspace{1.5cm} $r_\omega=|h(\omega)|$
\State \hspace{1cm} \textbf{end if}
\State \hspace{0.5cm} \textbf{end for}
\State \textbf{end for}
\State \textbf{for} $\omega$ \textbf{in} $T_{m}$ \textbf{do}
\State \hspace{0.5cm} $P_\omega \leftarrow C_0^\dagger P_\omega C_0$ and phase $\phi_{\omega_0} \in \{ \pm1 \}$
\State \textbf{end for}
\State Return $d_\omega =(\prod_{i=1}^m d_{\omega_i}\phi_{\omega_i})\phi_{\omega_0}\langle\braket{0|P_\omega}\rangle$ and $r_\omega$ for all $\omega$ in $T_{m}$
\end{algorithmic}
\textbf{end Procedure}
\end{algorithm}

As described in Theorem \ref{th:mainTheorem}, it is sufficient to have access to the quantity $\min_{\omega /\#\omega\geq \ell} |h(\omega)|$ to get a certificate on the $L^2$-error
made by our approximation. This quantity, which Algorithm \ref{alg:Algorithm} computed a lower bound to, is expected to scale linearly with $\ell$ in most cases. We present numerical evidence of such scaling by implementing our Algorithm on common PQC instances, namely QAOA for random 3-regular graphs. For each of these graph instances and associated QAOA quantum circuit, we backpropagate the observable $Z_iZ_j$, keeping track in the resulting tree of the quantity $r\leq\min_{\omega /\#\omega\geq \ell} |h(\omega)|$ for different values of $\ell$. The results are plotted in Figure $\ref{fig:expected_r}$.

\begin{figure}[h]
    \centering
    \includegraphics[scale=0.6]{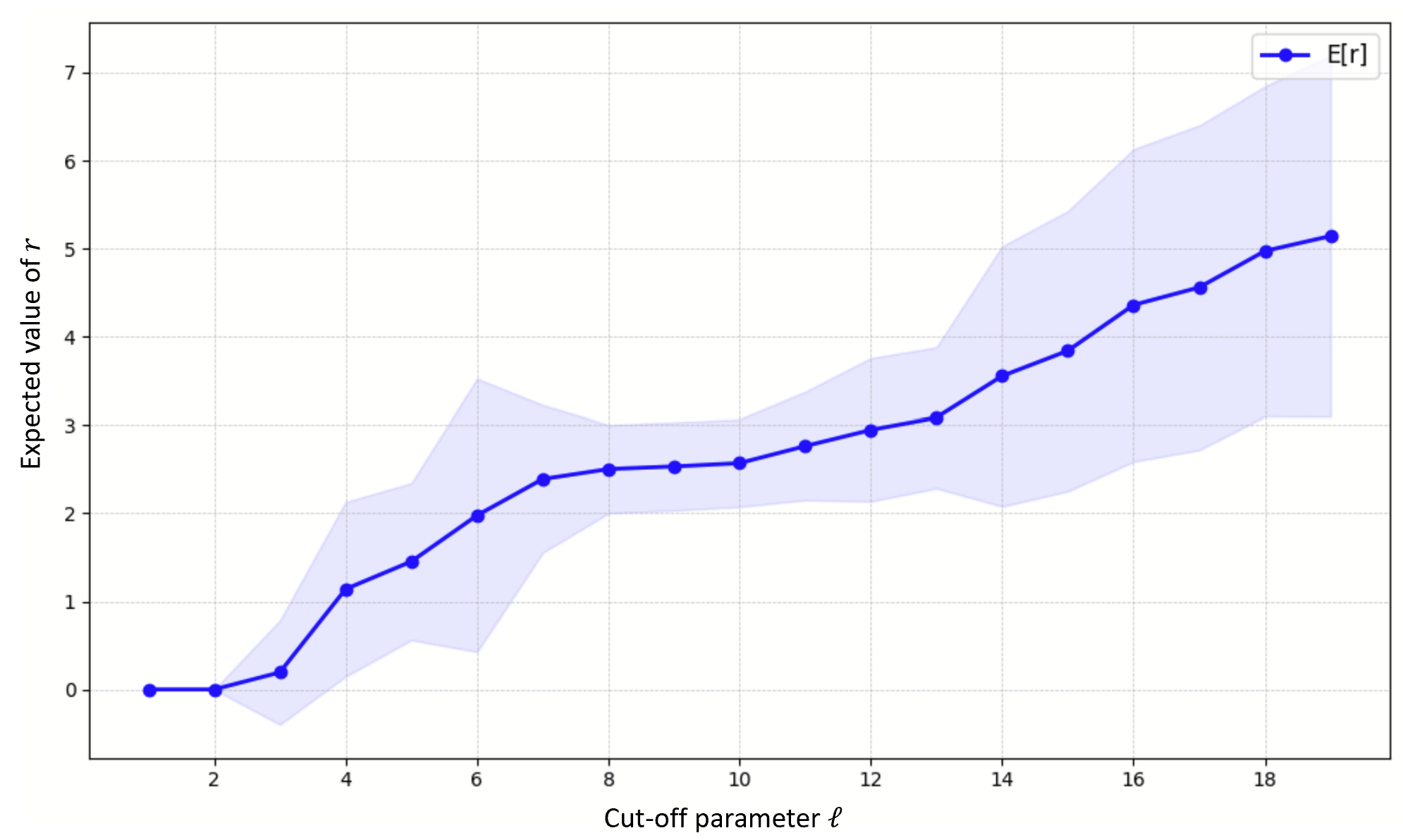}
    \caption{Lower bound $r$ on the expected number of splits through processes $\mathbf{D_{\pm1}}$ for a given cut-off parameter $\ell$. The quantum circuits considered are QAOA circuits on 3-regular graphs, with Pauli observables $Z_iZ_j$.}
    \label{fig:expected_r}
\end{figure}

\section{Proof of convergence}\label{sec:ProofConvergence}

Recall that the noisy circuits considered are made of alternating layers of Clifford unitary and single qubit $R_z$ rotations, such that,
\begin{equation}
    \mathcal{U_{\theta}}=\bigg(\bigcirc_{i=1}^m\mathcal{C}_i\circ \mathcal{R}^{(q_i)}_z(\theta_i)\circ\mathcal{N_{AD}}\bigg)\circ\mathcal{C}_0
\end{equation}

The noisy rotations can be decomposed in the PTM using the structure of the noise,

\begin{equation}
 \mathbf{R_z} \cdot \mathbf{N_{AD}}= \mathbf{D_0} + \mathbf{D_{0_Z}} + \mathbf{D_{0_I}} + \mathbf{D_{1}} \cos{\theta} + \mathbf{D_{-1}} \sin{\theta}
\end{equation}

where we have defined the quantum processes $\mathbf{D_0}=\ketbra{I}$, $\mathbf{D_{0_Z}}=(1-\gamma)\ketbra{Z}$, $\mathbf{D_{0_I}}=\gamma\ket{I}\bra{Z}$, $\mathbf{D_1}=\sqrt{1-\gamma}(\ketbra{X}+\ketbra{Y})$ and $\mathbf{D_{-1}}=\sqrt{1-\gamma}(\ket{Y}\bra{X}-\ket{X}\bra{Y})$, with $\ket{I} = [1 \ 0 \ 0 \ 0]^T, \ket{X}  = [0 \ 1 \ 0 \ 0]^T \ldots$ 
Our quantum circuit can then be further decomposed by summing all possible processes and keeping track of the rotation angles with the trigonometric function $\Phi_\omega(\theta)=\prod_{i=1}^m \phi_{\omega_i}(\theta_i)$, with $\phi_0(\theta_i)=\phi_{0_Z}(\theta_i)=\phi_{0_I}(\theta_i)=1$, $\phi_{1}(\theta_i)=\cos{\theta_i}$, and $\phi_{-1}(\theta_i)=\sin{\theta_i}$.

\begin{equation}
    \mathcal{U}_{\theta}=\sum_{\omega \in \{0,0_z,0_I,1,-1\}^m} \Phi_\omega(\theta)\mathcal{D}_\omega
\end{equation}

Each process $\mathcal{D}_\omega$ is defined in the Schr\"odinger picture by $\mathcal{D}_\omega = (\bigcirc_{i=1}\mathcal{C}_i \circ \mathcal{D}_{\omega_i}) \circ \mathcal{C}_0$. When backpropagating a Pauli string $P$, the rotation acts only on a single qubit $q_i$. Depending on the Pauli operator at qubit $q_i$, at most two processes will be valid: $\mathbf{D_0}$ for $I$, $\mathbf{D_{0_Z}}$ and $\mathbf{D_{0_I}}$ for $Z$, and $\mathbf{D_{\pm 1}}$ for $X$ and $Y$. This suggests the structure of a rooted tree $T = (V, E)$. The set of nodes $V$ represents the backpropagated Pauli strings at each layer, and the edges represent the valid processes $\mathbf{D_{\omega_i}}$ for the given Paulis. Our expectation value can then be rewritten as:

\begin{equation}
f(\theta) = \tr(\mathcal{U_{\theta}}[\ketbra{0}]P) = \sum_{\omega \in T} d_\omega \Phi_\omega(\theta)
\end{equation}

where the Fourier coefficients for each branch of the tree are defined by $d_\omega = \tr(\mathcal{D_{\omega}}[\ketbra{0}]P) = \langle\langle P|\mathbf{D_\omega}|0\rangle\rangle$. By backpropagating the Pauli string in this manner, we retrieve its expectation value, though at the cost of exponential computation time. The simulation requires updating the Clifford unitary for all paths of the tree at each layer. As there are $m$ layers and at most $2^m$ branches, the runtime of this exact simulation is at most $\mathcal{O}(n^2m2^m)$. Furthermore, as the Pauli observable is backpropagated through the quantum circuit, noise accumulates on each of the branches. Denote  $Q(\omega)$ the noise damping accumulated throughout the path $\omega$, with $Q(\omega)=\prod_i^m Q(\omega_i)$ where $Q(+1)=Q(-1)=\sqrt{1-\gamma}$, $Q(0)=1$, $Q(0_I)=\gamma$ and $Q(0_Z)=1-\gamma$. For some $d_\omega^0\in\{-1,0,1\}$, we can rewrite $d_\omega=Q(\omega)d_\omega^0$, such that,

\begin{equation}
f(\theta) = \tr(\mathcal{U_{\theta}}[\ketbra{0}]P) = \sum_{\omega \in T} Q(\omega)d_\omega^0 \Phi_\omega(\theta)
\end{equation}

Each time a branch in the tree splits, the resulting Paulis are damped by the noise. Once sufficient damping has occurred, the contribution of this branch to the expectation value can be neglected. We show that keeping only paths in the tree that have split fewer than a constant number of times $\ell$ yields a good approximation. More precisely, let $\#\omega$ denote the number of times a given path $\omega$ has split. We approximate the function $f$ by the quantity:

\begin{equation}
    \Tilde{f}(\theta) = \sum_{\substack{\omega \in T \\ \#\omega \leq \ell}} d_\omega \Phi_\omega(\theta)=\sum_{\substack{\omega \in T \\ \#\omega \leq \ell}} Q(\omega)d_\omega^0 \Phi_\omega(\theta)
\end{equation}

The metric used to judge the quality of our approximation error will be the $L^2$-norm between $f$ and $\Tilde{f}$ over the range of angles $\Theta=[0,2\pi]^m$.

\begin{equation}
    \begin{aligned}
        \Delta^2(\Tilde{f}, f) &= \frac{1}{|\Theta|}\int_\Theta |\Tilde{f}(\theta) - f(\theta)|^2 d\theta \\
        &= \frac{1}{|\Theta|}\int_{\Theta} | \sum_{\substack{\omega \in T \\ \#\omega > \ell}} \Phi_\omega(\theta) d_\omega |^2 d\theta
    \end{aligned}
\end{equation}

This error metric is convenient for several reasons. First, note that there is no hope of having a polynomial-time simulation algorithm under non-unital noise for all circuits, as celebrated results \cite{benor2013} have shown that it is possible to compute under such noise for an exponential amount of time and perform error correction. The $L^2$-norm error gives us convergence on average over the rotation angles, which does not contradict previous results since, for a given set of angles, the simulation might not converge. Furthermore, this simulation method creates a surrogate of the parametrized expectation value, which can then be used by plugging in different angles. The $L^2$-norm error therefore tells us how good our approximation is over the entire ``landscape''. Finally, when computing this metric, most of the paths will be orthogonal to each other, resulting in only a few terms remaining in the error which can then be analytically bounded.

\begin{lemma}\label{eq:lemmaProof}
    Denote $h$ a function that transforms a vector $\omega$ to a vector $h(\omega)$ where $0$, $0_Z$ and $0_I$ are assigned the same value $0$, the rest of the components being unchanged. Denote $\#\omega$ the number of times the path $\omega$ has split and $|h(\omega)|$ the number of non-zero elements of $h(\omega)$. The error made by the algorithm can be expressed as

    \begin{equation}
    \begin{aligned}
        \Delta^2(\Tilde{f},f) &= \sum_{\substack{(\omega,\omega')\in T^2\\
        h(\omega)=h(\omega') \\ \#\omega>\ell, \#\omega'>\ell}}Q(\omega)Q(\omega')2^{-|h(\omega)|}d_\omega^0 \Bar{d_{\omega'}^0}\\
    \end{aligned}
    \end{equation}
\end{lemma}

\begin{proof}
    Using known properties on the integral of product of trigonometric functions, one can compute the useful quantity: 
    \begin{equation}
        \frac{1}{|\Theta|}\int_\Theta \Phi_\omega(\theta)\Phi_{\omega'}(\theta)d\theta=\frac{1}{|\Theta|}\int_\Theta \prod_{i=1}^m \phi_{\omega_i}(\theta_i)\prod_{i=1}^m \phi_{\omega'_i}(\theta_i)d\theta=2^{-|h(\omega)|}\delta_{h(\omega)=h(\omega')}
    \end{equation}

    The approximation made by our algorithm can therefore be rewritten as:

    \begin{equation}
    \begin{aligned} \Delta^2(\Tilde{f},f)=\frac{1}{|\Theta|}\int_{\Theta}\big |\sum_{\substack{\omega \in T \\\#\omega>\ell}}\Phi_\omega(\theta)d_\omega \big |^2d\theta=\sum_{\substack{(\omega,\omega')\in T^2\\h(\omega)=h(\omega') \\ \#\omega>\ell, \#\omega'>\ell}}2^{-|h(\omega)|}d_\omega \Bar{d_{\omega'}}
    \end{aligned}
    \end{equation}

Finally, factoring out the noise $Q(\omega)$ and $Q(\omega')$ from the paths $\omega$ and $\omega'$ gives us the aforementioned equality.

\end{proof}

By definition of $d_\omega^0$, it is clear that for all $\omega$, $|d_\omega^0|\leq1$. This allows us to derive a simple bound on the error made:

\begin{equation}
    \Delta^2(\Tilde{f},f)\leq \sum_{\substack{(\omega,\omega')\in T^2\\h(\omega)=h(\omega') \\ \#\omega>\ell, \#\omega'>\ell}}Q(\omega)Q(\omega')2^{-|h(\omega)|}
\end{equation}

Our approximation error depends on the sum of pairs of paths $(\omega,\omega')$ of weights greater than $\ell$ satisfying some constraints. To quantify the contribution of each pair of paths, we introduce the quantity $v(\omega,\omega')$, that counts the number of positions where either $\omega$ or $\omega'$ or both have split, see Figure \ref{fig::backprog_v}. The error can clearly be bounded by:

\begin{equation}
    \Delta^2(\Tilde{f},f)\leq \sum_{\substack{(\omega,\omega')\in T^2\\h(\omega)=h(\omega') \\ v(\omega,\omega')>\ell}}Q(\omega)Q(\omega')2^{-|h(\omega)|}
\end{equation}

\begin{figure}
    \centering
    \begin{subfigure}{0.395\textwidth}
        \includegraphics[scale=0.37]{Whole_paths.png}
        \caption{}
    \end{subfigure}
    \begin{subfigure}{0.595\textwidth}
        \includegraphics[scale=0.68]{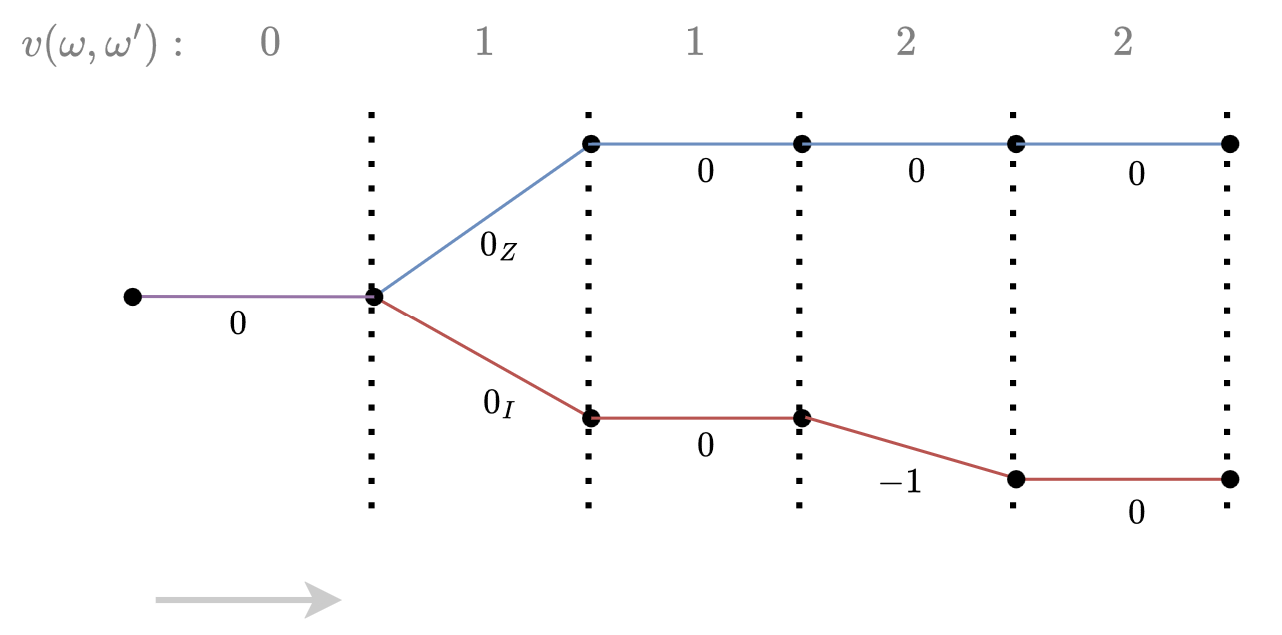}
        \caption{}.
    \end{subfigure}
    \caption{(a) Tree corresponding to the backpropagation of a single observable through the circuit. (b) Example of a pair of paths $(\omega,\omega')$ such that $v(\omega,\omega')=2$.}
    \label{fig::backprog_v}
\end{figure}

We show that splits that occur through processes $\mathbf{D_{\pm1}}$ dampen the approximation error by $1-\gamma$ and splits that occur through the pair of processes $\mathbf{D_{0_Z}}/\mathbf{D_{0_I}}$ leave the approximation error unchanged. Since random Clifford maps non-identity Paulis at random, adding a single-qubit random Clifford gate before the rotation allows us to derive the following bound:

\begin{prop}
Consider quantum circuits as defined previously, i.e. made of alternating layers of Clifford unitaries and 1-qubit rotations followed by a single random Clifford gate. The approximation error made by Algorithm \ref{alg:Algorithm} with cut-off $\ell$ is bounded by:
    \begin{equation}\label{eq:randomBound}
        \mathbb{E}[\Delta^2(\Tilde{f},f)]\leq [1-2\gamma/3]^{\ell+1},
    \end{equation}
    where the expectation is over the choice of the single-qubit random Clifford gates.
\end{prop}

\begin{proof}

For $i \in \{0, \dots, m\}$ and $\omega_{1,\ldots,i}$ and $\omega'_{1,\ldots,i}$ such that $h(\omega_{1,\ldots,i}) = h(\omega'_{1,\ldots,i})$, we define
    \begin{align*}
        \alpha_{j}(\omega_{1,\ldots,i},\omega'_{1,\ldots,i}) = \sum_{\substack{\omega_{i+1,\ldots,m},\omega'_{i+1,\ldots,m} \\ h(\omega) = h(\omega') \\
        v(\omega_{i+1,\ldots,m},\omega'_{i+1,\ldots,m}) \geq j}} Q(\omega_{i+1,\ldots,m}) Q(\omega_{i+1,\ldots,m}') 2^{-|h(\omega_{i+1,\ldots,m})|}
    \end{align*}
    Note that the quantity we wish to bound is $\alpha_{\ell+1}$ with the empty word ($i=0$). We show by induction on $j$ that for any $i$ and $\omega_{1,\ldots,i},\omega'_{1,\ldots,i}$, we have $\alpha_{j}(\omega_{1,\ldots,i},\omega'_{1,\ldots,i}) \leq [1 - 2\gamma/3]^{j}$. 
    
    The initialisation $j=0$ is straightforward as the expectation value of the Pauli string is bounded by 1. Assume now that the property is satisfied for $j$ and notice that both $Q$ and $2^{-|h|}$ can be expressed as products. We introduce $i'$ which corresponds to the index of the next split after $i$. Note that $\omega_{1,\ldots,i'-1}$ is the (unique) path extending $\omega_{1,\ldots,i}$ up to $i'-1$ similarly for $\omega'$. In addition, we have $Q(\omega_{i+1,\ldots ,i'-1})=Q(\omega_{i+1,\ldots ,i'-1})=1$ and $2^{-|h(\omega_{i+1,\ldots ,i'-1})|}=1$ which allows us to write,
    
     \begin{align*}
         &\alpha_{j+1}(\omega_{1,\ldots,i},\omega'_{1,\ldots,i}) 
     = \sum_{\substack{\omega_{i+1,\ldots,m},\omega'_{i+1,\ldots,m} \\ h(\omega) = h(\omega') \\
         v(\omega_{i+1,\ldots,m},\omega'_{i+1,\ldots,m}) \geq j+1}} Q(\omega_{i+1,\ldots,m}) Q(\omega'_{i+1,\ldots,m}) 2^{-|h(\omega_{i+1,\ldots,m})|} \\
    &= \sum_{\substack{\omega_{i+1,\ldots,i'},\omega'_{i+1,\ldots,i'} \\ h(\omega) = h(\omega')}} Q(\omega_{i+1,\ldots,i'}) Q(\omega'_{i+1,\ldots,i'}) 2^{-|h(\omega_{i+1,\ldots,i'})|}\sum_{\substack{\omega_{i'+1,\ldots,m},\omega'_{i'+1,\ldots,m} \\ h(\omega) = h(\omega')\\ 
     v(\omega_{i'+1,\ldots,m},\omega'_{i'+1,\ldots,m}) \geq j}} Q(\omega_{i'+1,\ldots,m}) Q(\omega'_{i'+1,\ldots,m}) 2^{-|h(\omega_{i'+1,\ldots,m})|}\\
     &= \sum_{\substack{\omega_{i+1,\ldots,i'},\omega'_{i+1,\ldots,i'} \\ h(\omega) = h(\omega')}} Q(\omega_{i+1,\ldots,i'}) Q(\omega'_{i+1,\ldots,i'}) 2^{-|h(\omega_{i+1,\ldots,i'})|}        \alpha_{j}(\omega_{1,\ldots,i'},\omega'_{1,\ldots,i'})
     \end{align*}

Using the recursive hypothesis on $\alpha_{j}(\omega_{1,\ldots,i'},\omega'_{1,\ldots,i'})$,

\begin{equation}
    \alpha_{j+1}(\omega_{1,\ldots,i},\omega'_{1,\ldots,i}) \leq [1 - 2\gamma/3]^{j} \sum_{\substack{\omega_{i+1,\ldots,i'},\omega'_{i+1,\ldots,i'} \\ h(\omega) = h(\omega')}} Q(\omega_{i+1,\ldots,i'}) Q(\omega'_{i+1,\ldots,i'}) 2^{-|h(\omega_{i+1,\ldots,i'})|}
\end{equation}

It remains to bound the sum for pairs of path $(\omega_{i+1,\ldots,i'},\omega'_{i+1,\ldots,i'})$. Note that by construction, there is at most one split occuring on either or both of these paths. We consider two distinct cases, which we will show lead to the same bound. The first one is the case where $\omega_{i+1,\ldots,i'-1}=\omega_{i+1,\ldots,i'-1}'$ and the second the one where $\omega_{i+1,\ldots,i'-1}\neq\omega_{i+1,\ldots,i'-1}'$ but $h(\omega_{i+1,\ldots,i'-1})=h(\omega_{i+1,\ldots,i'-1}')$.

  \begin{enumerate}
        \item \underline{$\omega_{i+1,\ldots,i'-1}=\omega_{i+1,\ldots,i'-1}'$:} In this case, illustrated in Figure \ref{fig:Xsplit}, it is clear that the same split occurs on both paths. If the split occurs through processes $\mathbf{D_{\pm1}}$ at $i'$, then the two resulting branches can only interact with themselves in order to satisfy $h(\omega)=h(\omega')$, and have contribution $Q(+1)Q(+1)2^{-{|h(+1)|}}+Q(-1)Q(-1)2^{-{|h(-1)|}}=(1-\gamma)$ in the sum. If the split occurs through $\mathbf{D_{0_Z}}/\mathbf{D_{0_I}}$, then there will be four terms in the sum as all pairs of paths will satisfy $h(\omega)=h(\omega')$. The result of the sum will be $Q(0_Z)Q(0_Z)2^{-{|h(0_Z)|}}+Q(0_I)Q(0_I)2^{-{|h(0_I)|}}+Q(0_Z)Q(0_I)2^{-{|h(0_Z)|}}+Q(0_I)Q(0_Z)2^{-{|h(0_I)|}}=\gamma^2+2\gamma(1-\gamma)+(1-\gamma)^2=1$. Since random Clifford gates map the non-identity Paulis uniformly at random between $\{X, Y, Z\}$, the probability of a split occuring through $\mathbf{D_{\pm1}}$ is $2/3$ and $1/3$ through $\mathbf{D_{0_Z}}/\mathbf{D_{0_I}}$ This allows us to compute the expected contribution of all the paths after the split:

        \begin{equation}
               \mathbb{E}[\sum_{\substack{\omega_{i+1,\ldots,i'},\omega'_{i+1,\ldots,i'} \\ h(\omega) = h(\omega')}} Q(\omega_{i+1,\ldots,i'}) Q(\omega'_{i+1,\ldots,i'}) 2^{-|h(\omega_{i+1,\ldots,i'})|}]\leq[1-2\gamma/3]
        \end{equation}

    \item \underline{$\omega_{i+1,\ldots,i'-1}\neq\omega_{i+1,\ldots,i'-1}'$:} In this case, illustrated in Figure \ref{fig:XYsplit}, we need to be more careful as the split occuring on each branch might be different, and two cases are possible. Either both the branches split (the splits might not be the same here), or only one splits, the other encountering process $\mathbf{D_0}$. Let us consider the latter case first. If one of the branches do not split, then maintaining $h(\omega)=h(\omega')$ requires the other branch splitting through $\mathbf{D_{0_Z}}/\mathbf{D_{0_I}}$. Once again, because of the random Clifford, this occurs with probability $1/3$, and leads to a sum of $Q(0)Q(0_Z)2^{-{|h(0)|}}+Q(0)Q(0_I)2^{-{|h(0)|}}=1\times \gamma+ 1\times (1-\gamma)=1$.

    If both of the branches split, we need to differientate based on the splits occuring. If one of the branches split through $\mathbf{D_{\pm1}}$ and the other through $\mathbf{D_{0_Z}}/\mathbf{D_{0_I}}$, then $h(\omega)\neq h(\omega')$. If both split through $\mathbf{D_{\pm1}}$, then the 4 new branches will lead to two terms satisfying $h(\omega)=h(\omega')$, carrying a damping of $Q(+1)Q(+1)2^{-{|h(+1)|}}+Q(-1)Q(-1)2^{-{|h(-1)|}}=(1-\gamma)$. Finally, if both split through processes $\mathbf{D_{0_Z}}/\mathbf{D_{0_I}}$, then the four terms arising are once again $Q(0_Z)Q(0_Z)2^{-{|h(0_Z)|}}+Q(0_I)Q(0_I)2^{-{|h(0_I)|}}+Q(0_Z)Q(0_I)2^{-{|h(0_Z)|}}+Q(0_I)Q(0_Z)2^{-{|h(0_I)|}}=\gamma^2+2\gamma(1-\gamma)+(1-\gamma)^2=1$. Because of the random Clifford, the probability of both split occuring through $\mathbf{D_{0_Z}}/\mathbf{D_{0_I}}$ is less than the probability of one of the paths splitting through $\mathbf{D_{\pm1}}$, \textit{i.e.} less than $1/3$. Once again, the expected damping is upper bounded by $[1-2\gamma/3]$. This bound also holds for the case where only one of the branches split such that,

        \begin{equation}
               \mathbb{E}[\sum_{\substack{\omega_{i+1,\ldots,i'},\omega'_{i+1,\ldots,i'} \\ h(\omega) = h(\omega')}} Q(\omega_{i+1,\ldots,i'}) Q(\omega'_{i+1,\ldots,i'}) 2^{-|h(\omega_{i+1,\ldots,i'})|}]\leq[1-2\gamma/3]
        \end{equation}

    \end{enumerate}

    Everytime a split occurs on a branch, its expected total contribution is damped by a factor $[1-2\gamma/3]$ which concludes the recursive argument and the proof.

    \begin{equation}
        \mathbb{E}[\Delta^2(f,\Tilde{f})]\leq[1-2\gamma/3]^{\ell+1}
    \end{equation}
    
    Obtaining Equation \ref{eq:certificateBound} of Theorem \ref{th:mainTheorem} can be done using the same proof technique. We've seen that splits occuring through processes $\mathbf{D_{\pm1}}$ dampen the approximation error by $1-\gamma$ and that splits occuring through processes $\mathbf{D_{0_Z}}/\mathbf{D_{0_I}}$ leave it unchanged. By counting the minimum number of times all discarded paths have split through processes $\mathbf{D_{\pm1}}$, which we denote $r$, we get that,

       \begin{equation}
        \Delta^2(f,\Tilde{f})\leq  (1-\gamma)^r
    \end{equation}

    Note that we expect $r$ to scale linearly in $\ell$ for most circuits, as shown in Figure \ref{fig:expected_r} for QAOA~\cite{farhi2014} instances on 3-regular graphs.
    \end{proof}

\section{Monte-Carlo Pauli backpropagation algorithm}\label{sec:MonteCarlo}
\begin{algorithm}
\caption{[MC-LOWESA-AD] Monte-Carlo simulation of Pauli string observables under non-unital noise}\label{alg:AlgorithmMC}
\textbf{Input:} Quantum circuit affected by non-unital noise, and defined by process modes $\{\mathcal{D}_\omega\}$ and noise coefficient $\gamma$; measurement Pauli operator $P$; cut-off parameter $\ell$; sampling overhead $K$\\
\textbf{Output:} $\Tilde{\hat{f}}(\theta)$, an approximation of $f(\theta)$

\textbf{Procedure} MC-LOWESA-AD: Given a quantum circuit with rotation parameters $\theta$ and Pauli observable $P$, outputs an approximation of the expectation value of $P$, by randomly sampling $K$ trees and creating a surrogate of the energy landscape for all those trees. The angles are then adjusted.
\begin{algorithmic}[1]
\State $\Tilde{\hat{f}}(\theta) \leftarrow 0$
\State \textbf{for} $k=1$ \textbf{to} $K$ \textbf{do}
\State \hspace{0.5cm}Run \emph{MC Low-weight coefficients} to calculate $d_\omega$ for all $\omega$ such that $\#\omega \leq \ell$.
\State \hspace{0.5cm}$\Tilde{f_k}(\theta) \leftarrow 0$
\State \hspace{0.5cm}\textbf{for all} $\omega$ such that $\#\omega \leq \ell$ with non-zero $d_\omega$ \textbf{do}
\State \hspace{1.cm} $\Tilde{f_k}(\theta) \leftarrow \Tilde{f_k}(\theta) + d_\omega \Phi_\omega (\theta)$
\State \hspace{0.5cm}\textbf{end for}
\State \hspace{0.5cm}$\Tilde{\hat{f}}(\theta) \leftarrow \frac{1}{K}\Tilde{f_k}(\theta)$
\State \textbf{end for}
\State Return $\Tilde{\hat{f}}(\theta)$
\end{algorithmic}
\textbf{end Procedure}\\

\textbf{Procedure} \emph{MC Low-weight coefficients}: Given a quantum circuit and Pauli observable $P$, outputs a randomly sampled backpropagated tree and its coefficients $d_\omega$ for all $\omega$ such that $\#\omega \leq \ell$ in time $\mathcal{O}(n^2m2^\ell)$. The sampled tree is constructed as follow: at each layer of the quantum circuit, the tree $T_i$ is extended at each branch $\omega$ based on the Clifford layer $C_{m-i+1}$ and the Pauli operator on the rotation qubit $q_{m-i+1}$. If $P^{q_{m-i+1}}_\omega=Z$, only one of the two branches is kept, with probability $\gamma$ or $1-\gamma$, which does not increase the total number of branches in the tree. Note that the initial tree $T_0$ is made of a single node corresponding to the Pauli observable $P$. The branches of the tree that have split more than $\ell$ times are discarded. 
\begin{algorithmic}[1]
\State \textbf{for} $i = 1$ \textbf{to} $m$ \textbf{do}
\State \hspace{0.5cm} \textbf{for} $\omega$ \textbf{in} $T_{i-1}$ \textbf{do}
\State \hspace{1cm} $P_\omega \leftarrow C_{m-i+1}^\dagger P_\omega C_{m-i+1}$, up to a phase $\phi_{\omega_i} \in \{ \pm1 \}$ which is stored.
\State \hspace{1cm} \textbf{if} $P_\omega^{q_{m-i+1}} = I$
\State \hspace{1.5cm} Add branch $\omega$ in $T_{i}$ with $\omega_i = 0$
\State \hspace{1.5cm} $d_{\omega_i}, P_\omega \leftarrow \mathcal{D}_{\omega_i}^\dagger (P_\omega)$
\State \hspace{1cm} \textbf{else if} $P_\omega^{q_{m-i+1}} = Z$ \textbf{ then}
\State \hspace{1.5cm} \textbf{with probability} $\gamma$ \textbf{do}
\State \hspace{2cm} Add branch $\omega$ in $T_{i}$ with $\omega_i = 0_I$
\State \hspace{2cm} $d_{\omega_i}, P_\omega \leftarrow \mathcal{D}_{\omega_i}^\dagger (P_\omega)$
\State \hspace{1.5cm} or\textbf{ with probability} $1-\gamma$ \textbf{do}
\State \hspace{2cm} Add branch $\omega'$ in $T_{i}$ with $\omega'_i = 0_Z$
\State \hspace{2cm} $d_{\omega'_i}, P_\omega' \leftarrow \mathcal{D}_{\omega'_i}^\dagger (P_\omega)$
\State \hspace{1.cm} \textbf{else if} $P_\omega^{q_{m-i+1}} \in \{X,Y\}$ and $\#\omega < \ell$ \textbf{then}
\State \hspace{1.5cm} Split into two branches $\omega$ and $\omega'$ in $T_{i}$ with $\omega_i= +1$ and $\omega'_i = -1$
\State \hspace{1.5cm} $d_{\omega_i}, P_\omega \leftarrow \mathcal{D}_{\omega_i}^\dagger (P_\omega)$
\State \hspace{1.5cm} $d_{\omega'_i}, P_{\omega'} \leftarrow \mathcal{D}_{\omega'_i}^\dagger (P_\omega)$
\State \hspace{1.5cm} $\#\omega,\#\omega' \leftarrow +1$
\State \hspace{1cm} \textbf{else if} $P_\omega^{q_{m-i+1}} \in \{X,Y\}$ and $\#\omega =\ell$ \textbf{then}
\State \hspace{1.5cm} Discard path $\omega$
\State \hspace{1cm} \textbf{end if}
\State \hspace{0.5cm} \textbf{end for}
\State \textbf{end for}
\State \textbf{for} $\omega$ \textbf{in} $T_{m}$ \textbf{do}
\State \hspace{0.5cm} $P_\omega \leftarrow C_0^\dagger P_\omega C_0$ and phase $\phi_{\omega_0} \in \{ \pm1 \}$
\State \textbf{end for}
\State Return $d_\omega =(\prod_{i=1}^m d_{\omega_i}\phi_{\omega_i})\phi_{\omega_0}\langle\braket{0|P_\omega}\rangle$ for all $\omega$ in $T_{m}$
\end{algorithmic}
\textbf{end Procedure}
\end{algorithm}

For \textit{most} quantum circuits, we have established that an extension of the standard LOWESA \cite{fontana2023} algorithm allows to approximate any observable in polynomial time in the $L^2$-norm. Provided one is willing to pay an additional sampling toll, we show through Algorithm \ref{alg:AlgorithmMC} that a Monte-Carlo approach allows us to obtain classical simulability for \textit{any} circuit. Using the tools previously introduced, this approach is straightforward to explain.

Consider the complete tree $T$ created by the backpropagation of the Pauli string $P$ through the quantum circuit. As seen previously, splits that occur in this tree due to processes $\mathbf{D_{\pm1}}$ dampen the expectation value of $P$, whereas splits occurring through processes $\mathbf{D_{0_Z}}$ and $\mathbf{D_{0_I}}$ might leave the norm of the expectation value unchanged, in the case of amplitude damping noise alone. However, these splits still create additional branches, increasing the computational cost of the simulation. 

To alleviate this problem, we randomly sample only one of these processes at each split. The process $\mathbf{D_{0_Z}}$ is kept with probability $1-\gamma$ and $\mathbf{D_{0_I}}$ with probability $\gamma$. To maintain an unbiased estimator, we make sure to remove the associated noise factors $\gamma$ or $1-\gamma$ from $d_\omega$. As a result, the nonphysical trees generated by this sampling approach only contain splits that effectively dampen the expectation value of $P$ by $\sqrt{1-\gamma}$. We can then approximate the trees by truncating paths that have split excessively.

We show that by repeatedly sampling nonphysical subtrees and truncating them at a given depth $\ell$, we obtain an efficient simulation algorithm.

\begin{thm}
    Let $\mathcal{U}_\theta$ be a noisy quantum circuit as defined in Equation \ref{eq:unitaryLOWESA} and $P$ a Pauli observable with expectation value $f(\theta)=\tr(\mathcal{U_{\theta}}[\ketbra{0}]P)$. It is possible to compute an approximation $\Tilde{\hat{f}}$ of $f$ in time $\mathcal{O}(Kn^2m2^\ell)$, where $\ell$ and $K$ are a chosen cut-off parameter and sampling overhead, such that with probability at least $1-\delta$,

\begin{equation}
        \Delta(f,\hat{\Tilde{f}})\leq  (1-\gamma)^{(\ell+1)/2}+ \sqrt{\frac{2\log{(\delta^{-1}/2})}{K}}
    \end{equation}
\end{thm}
where the probability comes from the random sampling of the trees.
\begin{proof}

Denote by $T_k$ a nonphysical tree created by sampling processes $\mathbf{D_{0_Z}}$ with probability $1-\gamma$ or $\mathbf{D_{0_I}}$ with probability $\gamma$ each time a Pauli $Z$ is encountered by the noisy rotation when backpropagating the observable. It is straightforward to see that the following estimator is unbiased:
\begin{equation}
    \hat{f}(\cdot) = \frac{1}{K}\sum_{k=1}^K\sum_{\omega \in T_k} \Phi_\omega(\cdot) d_\omega.
\end{equation}

Indeed, consider a split occurring through processes $\mathbf{D_{0_Z}}$ and $\mathbf{D_{0_I}}$. This gives rise to two different branches, whose contributions to $f$ we denote by $(1-\gamma) \times S_1$ and $\gamma \times S_2$, where the damping prefactor comes from the processes. By randomly sampling only one of the branches with the probabilities mentioned, it is clear that the expectation value of this estimator is $(1-\gamma)\times S_1 + \gamma\times S_2$, making it unbiased.

Instead of considering the entire trees created by this Monte Carlo approach, which may still have an exponential number of branches, we propose only taking into account the branches of each tree that have split fewer than a constant number of times $\ell$. This approximation of our unbiased estimator is written as:

    \begin{equation}
        \Tilde{\hat{f}}(\cdot)=\frac{1}{K}\sum_{k=1}^K\sum_{\substack{\omega\in T_k\\\#\omega\leq\ell}}\Phi_\omega(\cdot)d_\omega
    \end{equation}

Note that for the trees $T_k$ considered, all splits come from processes $\mathbf{D_{\pm 1}}$. The error we wish to bound is the expected $L^2$-error between our original parametrized expectation value $f$, and the one obtained by the truncated Monte Carlo approach $\Tilde{\hat{f}}$.

\begin{equation}
    \mathbb{E}[\Delta(f,\Tilde{\hat{f}})] = \mathbb{E}[||f-\Tilde{\hat{f}}||_2] = \mathbb{E}\left[\left(\frac{1}{|\Theta|}\int_{\Theta} |f(\theta) - \Tilde{\hat{f}}(\theta)|^2 d\theta\right)^{1/2}\right]
\end{equation}

This error can be upper bounded by two distinct errors using the triangle inequality: one corresponding to the truncation error and the other to the Monte Carlo approximation, $||f-\Tilde{\hat{f}}||_2 = ||(f-\hat{f}) + (\hat{f}-\Tilde{\hat{f}})||_2\leq||f-\hat{f}||_2 + ||\hat{f}-\Tilde{\hat{f}}||_2$. 

\begin{equation}
    \mathbb{E}[\Delta(f,\Tilde{\hat{f}})] \leq \mathbb{E}[||f-\hat{f}||_2] + \mathbb{E}[||\hat{f}-\Tilde{\hat{f}}||_2]
\end{equation}

The first term comes from the mean square error of our Monte Carlo simulation, and is bounded by $1/K$, as our estimator is unbiased and the sampled trees yield expectation values in the range of $[-1,1]$. Computing the second term, the error made by truncating each of our simulated trees, requires additional effort.

\begin{equation}
    ||\hat{f}-\Tilde{\hat{f}}||_2^2 = \frac{1}{|\Theta|} \int_{\Theta} |\hat{f}(\theta) - \Tilde{\hat{f}}(\theta)|^2 d\theta = \frac{1}{K^2} \sum_{k=1}^K \sum_{k'=1}^K \frac{1}{|\Theta|} \int_{\Theta} \sum_{\substack{\omega \in T_k \\ \#\omega > \ell}} \Phi_\omega(\theta) d_\omega \sum_{\substack{\omega' \in T_{k'} \\ \#\omega' > \ell}} \Phi_{\omega'}(\theta) \bar{d_{\omega'}} d\theta
\end{equation}

Using the Cauchy-Schwarz inequality, Lemma \ref{eq:lemmaProof} (since the only splits in the trees are from $\mathbf{D_{\pm1}}$, the condition $h(\omega)=h(\omega')$ can be reduced to $\omega=\omega'$) and the fact that all splits in a tree $T_k$ come from processes $\mathbf{D_{\pm1}}$, the error term can be rewritten:

\begin{equation}
\begin{aligned}
    ||\hat{f}-\Tilde{\hat{f}}||_2^2 &\leq \frac{1}{K^2}\sum_{k=1}^K\sum_{k'=1}^K \norm{\sum_{\substack{\omega \in T_k \\ \#\omega > \ell}}\Phi_\omega d_\omega}_2\norm{\sum_{\substack{\omega \in T_{k'} \\ \#\omega > \ell}}\Phi_\omega d_\omega}_2\\
    &\leq\frac{1}{K^2}\bigg(\sum_{k=1}^K\bigg(\sum_{\substack{(\omega,\omega')\in T_k^2\\
    \#\omega>\ell, \#\omega'>\ell\\
    \omega=\omega'}}2^{-\#\omega}d_\omega \bar{d_{\omega'}}\bigg)^{1/2}\bigg)^2 
\end{aligned}
\end{equation}

Once again, it is possible to rewrite $d_\omega=Q(\omega)d_\omega^0$ where $Q(\omega)$ represents the noise accumulated through path $\omega$, with $Q(\omega)=\prod_iQ(\omega_i)$ with $Q(+1)=Q(-1)=\sqrt{1-\gamma}$ and $Q(0)=1$. Since $d_\omega^0$ is still the expectation value of a single Pauli string, we have $|d_\omega^0|\leq 1$. The paths having split through the $\mathbf{D_{\pm1}}$ processes at least $\ell+1$ times,

\begin{equation}
    ||\hat{f}-\Tilde{\hat{f}}||_2^2 \leq \frac{1}{K^2}\bigg(\sum_{k=1}^K\bigg(\sum_{\substack{(\omega,\omega')\in T_k^2\\
    \#\omega>\ell, \#\omega'>\ell\\
    \omega=\omega'}}2^{-\#\omega}Q(\omega)^2|d_\omega^0|^2\bigg)^{1/2}\bigg)^2\leq\frac{(1-\gamma)^{\ell+1}}{K^2}\bigg(\sum_{k=1}^K\bigg(\sum_{\substack{(\omega,\omega')\in T_k^2\\
    \#\omega>\ell, \#\omega'>\ell\\
    \omega=\omega'}}2^{-\#\omega}\bigg)^{1/2}\bigg)^2
\end{equation}

Finally, it remains to notice that for any given tree $T_k$,

\begin{equation}
    \sum_{\substack{(\omega,\omega')\in T_k^2\\
    \#\omega>\ell, \#\omega'>\ell\\
    \omega=\omega'}}2^{-\#\omega}\leq\sum_{\substack{(\omega,\omega')\in T_k^2\\
    \omega=\omega'}}2^{-\#\omega}=1
\end{equation}

This allows us to bound the truncation error and shows,

\begin{equation}
    \mathbb{E}[\Delta(f,\Tilde{\hat{f}})] \leq  (1-\gamma)^{(\ell+1)/2}+1/\sqrt{K}
\end{equation}

By applying Hoeffding's inequality and noting that each of the sampled trees yields an expectation value in the range of $[-1,1]$, we can conclude that with probability at least $1-\delta$,

\begin{equation}
        \Delta(f,\hat{\Tilde{f}})\leq  (1-\gamma)^{(\ell+1)/2}+ \sqrt{\frac{2\log{(\delta^{-1}/2})}{K}}
\end{equation}
\end{proof}

\section{Extension to normal form non-unital noise channel}

The results previously obtained for the amplitude damping channel can be extended to any non-unital noise that can be brought into the \emph{normal form} representation of the noise channel~\cite{ruskai2001,benor2013,mele2024} by acting on the left or right with single qubit unitaries. More precisely, recall that for any single-qubit noise channel $\mathcal{N}$, there exist unitaries $U,V$  with the property that $\mathcal{U}\mathcal{N}\mathcal{V}$ can be represented by two vectors $\mathbf{t}=(t_X,t_Y,t_Z) \in \mathbb{R}^3$ and $\mathbf{D}=(D_X,D_Y,D_Z) \in \mathbb{R}^3$ such that its action on a vector of the Bloch sphere can be written,

\begin{equation}
    \mathcal{N}\left(\frac{I+\mathbf{w}\cdot\mathbf{\sigma}}{2}\right)=\frac{I}{2}+\frac{1}{2}(\mathbf{t}+D\mathbf{w})\cdot\mathbf{\sigma}
\end{equation}

where we have defined $D:=\diag(\mathbf{D})$, $\mathbf{\sigma}:=\{X,Y,Z\}$ the vector of single-qubit Pauli matrices and $\mathbf{w}\in\mathbb{R}^3$ such that $||\mathbf{w}||_2\leq1$. The adjoint channel $\mathcal{N}^\dagger$ acts on the Pauli basis in the following way,

\begin{equation}
    \forall P\in\{X,Y,Z\},\quad \mathcal{N}^\dagger(P)=t_PI+D_PP
\end{equation}
We will now consider channels that can be brought into this form where the unitaries $U,V$ are Clifford. Note that the special cases of depolarizing, dephasing and amplitude damping channels are members of this class with $U,V=I$. Furthermore, note that any composition or convex combination of such channels will remain in the class. As such, this class of channels cover probabilistic and composite combinations of the majority of single qubit channels considered in the literature.
We consider the case where the noise channel is non-unital, which implies $\mathbf{t}\neq(0,0,0)$~\cite{mele2024}. Note that unital noise is already known to be simulable by similar frequency truncation techniques~\cite{fontana2023}. We will be making extensive use of the following properties of the channel in its normal form.

\begin{lemma}[Constraints on the parameters of the normal form, adapted from Ref.\ \cite{angrisani2025simulating}] Let $\textbf{D},\textbf{t}$ be the parameters of a non-unital channel in normal form. Then, for any $\mathbf{b}=(b_X,b_Y,b_Z)\in\mathbb{R}^3$ such that $||\mathbf{b}||_2=1$, 
\begin{equation}\label{eq:constaint_normal}
    \sum_{P\in\{X,Y,Z\}}b_P^2|D_P|+\sum_{P\in\{X,Y,Z\}}|b_Pt_P|\leq 1
\end{equation}

\noindent In particular we have the two following properties:
\begin{enumerate}
    \item $\forall P \in \{X,Y,Z\},\quad |D_P|+|t_P|\leq1$,
    \item and there exists at most one $P \in \{X,Y,Z\}$ such that $ |D_P|+|t_P|=1$.
\end{enumerate}
\end{lemma}

\begin{proof}
    Let $O$ be the observable and $\rho$ the state defined by,

    \begin{equation}
        O=\sum_{P\in\{X,Y,Z\}}|b_P|\cdot\sign(t_PD_P)P \quad\text{ and } \quad\rho=\frac{I+O}{2}
    \end{equation}

    The operator norm of $O$ can be bounded as follow,

    \begin{equation}
        ||O||_\infty=\max_{\sigma}|\tr(O\sigma)|=\max_{\substack{\mathbf{r}\in\mathbb{R}^3 \\||\mathbf{r}||_2=1}}\left|\sum_{P\in\{X,Y,Z\}}|b_P|\cdot r_P\sign(t_P)\right|\leq\max_{\substack{\mathbf{r}\in\mathbb{R}^3 \\||\mathbf{r}||_2=1}}||\mathbf{b}||_2||\mathbf{r}||_2=1
    \end{equation}

Where the last inequality is obtained through Cauchy-Schwarz. Equality is attained by choosing $\sigma=\rho$, such that $||O||_\infty=1$. Furthermore, Hölder's inequality allows us to bound,

\begin{align}
    &1=||O||_\infty||\rho||_1\geq |\tr[O\mathcal{N}(\rho)]|\\
    &=\frac{1}{2}\left| \tr\left[  \left( \sum_{P\in\{ X,Y,Z\}}|b_P|\cdot\sign(t_PD_P)P \right)\mathcal{N}\left( I+\sum_{P\in\{X,Y,Z\}}|b_P|\cdot\sign(t_PD_P)P \right) \right] \right|\\
    &=\frac{1}{2}\left| \tr\left[  \left( \sum_{P\in\{ X,Y,Z\}}|b_P|\cdot\sign(t_PD_P)P \right)\left( I+\sum_{P\in\{X,Y,Z\}}(|b_P|D_P\cdot\sign(t_PD_P)+t_P)P \right) \right] \right|\\
    &=\left| \sum_{P\in\{X,Y,Z\}}b_P^2D_P+|b_Pt_P|\sign(D_P)\right|=    \sum_{P\in\{X,Y,Z\}}b_P^2|D_P|+\sum_{P\in\{X,Y,Z\}}|b_Pt_P|,
\end{align}
where in the last step we used the fact that the coefficients $D_P$ have all the same sign.
Obtaining the first property mentioned is then done by setting $\mathbf{b}$ to either $(1,0,0)$, $(0,1,0)$ or $(0,0,1)$, which shows that $\forall P \in \{X,Y,Z\},\quad |D_P|+|t_P|\leq1$. Suppose now that for one of the Paulis $P$, $|D_P|+|t_P|=1$. Then there cannot be another $P'\neq P$ such that $|D_{P'}|+|t_{P'}|=1$. Indeed, Equation \ref{eq:constaint_normal} would imply for any $0<b_P<1$ and $0<b_{P'}<1$ with $b_P^2+b_{P'}^2=1$ that,

\begin{equation}
    b_{P'}^2D_{P'}+b_P^2D_P+b_{P'}|t_{P'}|+b_P|t_P|\leq1
\end{equation}
However, this is not the case if both $|D_{P'}|+|t_{P'}|=1$ and $|D_P|+|t_P|=1$ and either $t_{P'}\neq0$ or $t_P\neq0$,
\begin{equation}
\begin{aligned}
    b_{P'}^2D_{P'}+b_P^2D_P+b_{P'}|t_{P'}|+b_P|t_P|&=
    b_{P'}^2(1-|t_{P'}|)+b_P^2(1-|t_P|)+b_Z|t_Z|+b_P|t_P|\\
    &=1+b_{P'}|t_{P'}|(1-b_{P'})+b_P|t_P|(1-b_P)\\
    &>1,
\end{aligned}
\end{equation}
where in the second step we used the fact that $b_{P'}^2+b_P^2=1$.
\end{proof}

 Note that if there is no $P$ such that  $|D_P|+|t_P|=1$, then it is possible to simulate circuits with rotations along \emph{any} axis, making our simulation method even more powerful.
In the following, if the noise channel is such that there exists a $P\in\{X,Y,Z\}$ satisfying $|D_P|+|t_P|=1$, then we will consider $P=Z$, similarly to the amplitude damping case.
This implies that $|D_X|+|t_X|<1$ and $|D_Y|+|t_Y|<1$. Going back to the circuit architecture considered, our noisy circuit is written,
\begin{equation}\label{eq:circuit_normalform}
    \mathcal{U_{\theta}}=\bigg(\bigcirc_{i=1}^m\mathcal{C}_i\circ \mathcal{R}^{(q_i)}_z(\theta_i)\circ\mathcal{N}\bigg)\circ\mathcal{C}_0
\end{equation}

We will now consider that the Cliffords putting in diagonal form are identity without loss of generality, as the additional diagonalizing Cliffords can be absorbed into the circuit. With the noise in its normal form, the noisy rotation can be written as

\begin{equation}
\begin{aligned}
\mathbf{R_z}\cdot\mathbf{N} 
&=\begin{bmatrix}
1 & t_X & t_Y & t_Z \\
0 & D_X\cos{\theta} & -D_Y\sin{\theta} &0 \\ 
0 & D_X\sin{\theta} & D_Y\cos{\theta} & 0 \\
0 & 0 & 0 & D_Z
\end{bmatrix}\\
 &=\mathbf{D_0} + \sum_{P\in\{X,Y,Z\}}\mathbf{D_{0_P}}+\sum_{P\in\{X,Y\}}(\mathbf{D_{1_P}} \cos{\theta} + \mathbf{D_{-1_P}} \sin{\theta})+\mathbf{D_{0_I}}
 \end{aligned}
\end{equation}

where we have defined the quantum processes $\mathbf{D_0}=\ketbra{I}$, $\mathbf{D_{0_P}}=t_P\ket{I}\bra{P}$, $\mathbf{D_{0_I}}=D_Z\ket{Z}\bra{Z}$, $\mathbf{D_{1_P}}=D_P \ketbra{P}$, $\mathbf{D_{-1_X}}=D_X\ket{Y}\bra{X}$ and $\mathbf{D_{-1_Y}}=-D_Y\ket{X}\bra{Y}$, with $\ket{I} = [1 \ 0 \ 0 \ 0]^T, \ket{X}  = [0 \ 1 \ 0 \ 0]^T \ldots$.

When backpropagating our Pauli string, the noisy rotation acts only on a single qubit $q_i$. Depending on the Pauli operator at qubit $q_i$, at most 3 processes are valid: $\mathbf{D_{0}}$ for $I$, $\mathbf{D_{0_Z}}$ and $\mathbf{D_{0_I}}$ for $Z$, $\mathbf{D_{0_X}}$, $\mathbf{D_{1_X}}$ and $\mathbf{D_{-1_X}}$ for $X$ and $\mathbf{D_{0_Y}}$, $\mathbf{D_{1_Y}}$ and $\mathbf{D_{-1_Y}}$ for $Y$. This once again suggests the structure  of a rooted tree $T = (V, E)$, where the nodes $V$ represent the backpropagated Pauli strings at each layer, and the edges represent the valid processes $\mathbf{D_{\omega_i}}$ for the given Paulis.
Similarly as in the amplitude damping case, keeping track of the rotation monomials and the splitting of the paths allows us to write our expectation value as

\begin{equation}
f(\theta) = \tr(\mathcal{U_{\theta}}[\ketbra{0}]P) = \sum_{\omega \in T} d_\omega \Phi_\omega(\theta)
\end{equation}

However, each branch can split into up to three different branches, which means that recovering $f$ exactly would require time $\mathcal{O}(n^2m3^m)$ in the worst case. To handle the additional branches coming from the non-unital noise, we once again propose to randomly sample only some of them. More particularly, our Monte Carlo approach creates nonphysical trees as follow:

\begin{itemize}
    \item If $P_{q_i}=X \text{ or } Y$, keep process $\mathbf{D_{0_P}}$ with probability $|t_P|/(|D_P|+|t_P|)$ and both processes $\mathbf{D_{\pm1_P}}$ with probability $|D_P|/(|D_P|+|t_P|)$.
    \item If $P_{q_i}=Z$, keep process $\mathbf{D_{0_Z}}$ with probability $|t_Z|/(|D_Z|+|t_Z|)$ and process $\mathbf{D_{0_I}}$ with probability $|D_Z|/(|D_Z|+|t_Z|)$
\end{itemize}

Note that keeping the estimator unbiased requires us to multiply the branches kept by $(|D_P|+|t_P|)$ at each split. For $P=X\text{ or }Y$, this effectively contract the expectation value since for $P\in\{X,Y\},\quad(|D_P|+|t_P|)<1$. Therefore, truncating the sampled trees after a constant amount of splits $\ell$ through $\mathbf{D_{\pm1_P}}$ yields a good approximation.

\begin{thm}
    Let $\mathcal{U}_\theta$ be a noisy quantum circuit as defined in Equation \ref{eq:circuit_normalform} and $P$ a Pauli observable with expectation value $f(\theta)=\tr(\mathcal{U_{\theta}}[\ketbra{0}]P)$. Denote by $|D|+|t|=\max_{P\in\{X,Y\}}|D_P|+|t_P|<1$. It is possible to compute an approximation $\Tilde{\hat{f}}$ of $f$ in time $\mathcal{O}(Kn^2m2^\ell)$, where $\ell$ and $K$ are a chosen cut-off parameter and sampling overhead, such that with probability at least $1-\delta$,
\begin{equation}
        \Delta(f,\hat{\Tilde{f}})\leq  (|D|+|t|)^{(\ell+1)}+ \sqrt{\frac{2\log{(\delta^{-1}/2})}{K}}
\end{equation}
     where the probability comes from the random sampling of the trees.
\end{thm}

\begin{proof}
Denote by $T_k$ a nonphysical tree created by sampling processes $\mathbf{D_{0_P}}$ with probability $|t_P|/(|D_P|+|t_P|)$ or keeping the other branche(s) (both $\mathbf{D_{\pm1_P}}$ if $P\in\{X,Y\}$ or $\mathbf{D_{0_I}}$ if $P=Z$) with probability $|D_P|/(|D_P|+|t_P|)$ each time a Pauli $P$ is encountered by the noisy rotation when backpropagating the observable. It is straightforward to see that the following estimator is unbiased:
\begin{equation}
    \hat{f}(\cdot) = \frac{1}{K}\sum_{k=1}^K\sum_{\omega \in T_k} \Phi_\omega(\cdot) d_\omega.
\end{equation}

Indeed, consider a branch that splits by encountering a Pauli $P\in\{X,Y\}$ during the backpropagation. This split leads to 3 different branches through processes $\mathbf{D_{0_P}}$, $\mathbf{D_{1_P}}$ and $\mathbf{D_{-1_P}}$, which contribution to the expectation value we can write $D_P \cos\theta S_1 + D_P \sin\theta S_2 +t_P S_3$, where the noise prefectors come from the processes. By randomly sampling the branches with the probabilities mentioned before and multiplying their contributions by the factor $|D_p|+|t_P|$, the expectation value of our estimator is $|t_P|/(|D_p|+|t_P|)\times S_3 \times (|D_P|+|t_P|) + |D_P|/(|D_p|+|t_P|)\times (\cos \theta S_1 +\sin\theta S_2)\times(|D_P|+|t_P|)$. To make this truly unbiased, one also needs to adjust for the sign of $t_P$ or $D_P$ when multiplying by $(|D_p|+|t_P|)$. Similarly, if the split comes from processes $\mathbf{D_{0_Z}}/\mathbf{D_{0_I}}$, the two resulting branches have contribution $D_Z \times S_1 + t_Z \times S_2$. Sampling either $\mathbf{D_{0_Z}}$ with probability $|t_Z|/(|D_Z|+|t_Z|)$ or $\mathbf{D_{0_Z}}$ with probability $|D_Z|/(|D_Z|+|t_Z|)$ and multiplying the sampled branch by $|D_Z|+|t_Z|$ yields an expectation value of $|D_Z|/(|D_Z|+|t_Z|)\times S_1\times(|D_Z|+|t_Z|) +|t_Z|/(|D_Z|+|t_Z|)\times S_2 \times(|D_Z|+|t_Z|)$. Adjusting once again for the sign of $t_Z$ and $D_Z$ suffices to make our estimator unbiased.

The trees $T_k$ created now only contain splits that occur through processes $\mathbf{D_{\pm1_X}}$ or $\mathbf{D_{\pm1_Y}}$. There may still be an exponential amount of branches in the tree though, so we complement this Monte-Carlo approach with the same truncation rule as before, and only keep branches that have split less than a fixed number of times $\ell$. At each of these splits we show that the error made by our approximation is dampen by $|D_X|+|t_X|<1$ or $|D_Y|+|t_Y|<1$.

This approximation of our unbiased estimator is written as:

    \begin{equation}
        \Tilde{\hat{f}}(\cdot)=\frac{1}{K}\sum_{k=1}^K\sum_{\substack{\omega\in T_k\\\#\omega\leq\ell}}\Phi_\omega(\cdot)d_\omega
    \end{equation}

 The error we wish to bound is once again the $L^2$-error between our original parametrized expectation value $f$, and the one obtained by the truncated Monte Carlo approach $\Tilde{\hat{f}}$. This error is still upper bounded by two distinct errors using the triangle inequality: one corresponding to the truncation error and the other to the Monte Carlo approximation, $||f-\Tilde{\hat{f}}||_2 = ||(f-\hat{f}) + (\hat{f}-\Tilde{\hat{f}})||_2\leq||f-\hat{f}||_2 + ||\hat{f}-\Tilde{\hat{f}}||_2$. 

The first term comes from the mean square error of our Monte Carlo simulation, and can be bounded using Hoeffding's inequality, as our estimator is unbiased. Computing the second term, the error made by truncating can be done in a similar fashion as for amplitude damping.

\begin{equation}
    ||\hat{f}-\Tilde{\hat{f}}||_2^2 = \frac{1}{|\Theta|} \int_{\Theta} |\hat{f}(\theta) - \Tilde{\hat{f}}(\theta)|^2 d\theta = \frac{1}{K^2} \sum_{k=1}^K \sum_{k'=1}^K \frac{1}{|\Theta|} \int_{\Theta} \sum_{\substack{\omega \in T_k \\ \#\omega > \ell}} \Phi_\omega(\theta) d_\omega \sum_{\substack{\omega' \in T_{k'} \\ \#\omega' > \ell}} \Phi_{\omega'}(\theta) \bar{d_{\omega'}} d\theta
\end{equation}

Using the Cauchy-Schwarz inequality, Lemma \ref{eq:lemmaProof} (since the only splits in the trees are from $\mathbf{D_{\pm1_X}}$ or $\mathbf{D_{\pm1_Y}}$ , the condition $h(\omega)=h(\omega')$ can be reduced to $\omega=\omega'$) and the fact that all splits in a tree $T_k$ come from processes $\mathbf{D_{\pm1_X}}$ or $\mathbf{D_{\pm1_Y}}$, the error term can be rewritten:

\begin{equation}
\begin{aligned}
    ||\hat{f}-\Tilde{\hat{f}}||_2^2 &\leq \frac{1}{K^2}\sum_{k=1}^K\sum_{k'=1}^K \norm{\sum_{\substack{\omega \in T_k \\ \#\omega > \ell}}\Phi_\omega d_\omega}_2\norm{\sum_{\substack{\omega \in T_k' \\ \#\omega > \ell}}\Phi_\omega d_\omega}_2\\
    &\leq\frac{1}{K^2}\bigg(\sum_{k=1}^K\bigg(\sum_{\substack{(\omega,\omega')\in T_k^2\\
    \#\omega>\ell, \#\omega'>\ell\\
    \omega=\omega'}}2^{-\#\omega}d_\omega \bar{d_{\omega'}}\bigg)^{1/2}\bigg)^2 
\end{aligned}
\end{equation}

 Since we have sampled the branches with probability $|t_P|/(|D_P|+|t_P|)$ and $|D_P|/(|D_P|+|t_P|)$, keeping the estimator unbiased required us to multiply each branch by $(|D_P|+|t_P|)$, and adjust for the sign of $t_P$ or $D_P$. It is therefore possible to rewrite for all $\omega\in T_k$, $d_\omega=Q(\omega)d_\omega^0$ where $Q(\pm1_X)=Q(0_X)=|D_X|+|t_X|<1$, $Q(\pm1_Y)=Q(0_Y)=|D_Y|+|t_Y|<1$, $Q(0_I)=Q(0_Z)=|D_Z|+|t_Z|\leq1$ and $|d_\omega^0|\leq 1$. Denote by $|D|+|t|=\max_{P\in\{X,Y\}}|D_P|+|t_P|$. Since each path in the sum has split at least $\ell+1$ times through processes $\mathbf{D_{\pm1_X}}$ or $\mathbf{D_{\pm1_Y}}$, we have that,

\begin{equation}    
    ||\hat{f}-\Tilde{\hat{f}}||_2^2 \leq \frac{1}{K^2}\bigg(\sum_{k=1}^K\bigg(\sum_{\substack{(\omega,\omega')\in T_k^2\\
    \#\omega>\ell, \#\omega'>\ell\\
    \omega=\omega'}}2^{-\#\omega}Q(\omega)^2|d_\omega^0|^2\bigg)^{1/2}\bigg)^2\leq\frac{(|D|+|t|)^{2(\ell+1)}}{K^2}\bigg(\sum_{k=1}^K\bigg(\sum_{\substack{(\omega,\omega')\in T_k^2\\
    \#\omega>\ell, \#\omega'>\ell\\
    \omega=\omega'}}2^{-\#\omega}\bigg)^{1/2}\bigg)^2
\end{equation}

Since the sampled trees contain only splits in up to two branches, and branches only split when encoutering processes $\mathbf{D_{\pm1_X}}$ or $\mathbf{D_{\pm1_Y}}$, we still have that for any given tree $T_k$,

\begin{equation}
    \sum_{\substack{(\omega,\omega')\in T_k^2\\
    \#\omega>\ell, \#\omega'>\ell\\
    \omega=\omega'}}2^{-\#\omega}\leq\sum_{\substack{(\omega,\omega')\in T_k^2\\
    \omega=\omega'}}2^{-\#\omega}=1
\end{equation}

This allows us to bound the truncation error,

\begin{equation}
    ||\Tilde{\hat{f}}-\hat{f}||_2\leq (|D|+|t|)^{\ell+1}
\end{equation}

Using Hoeffding's inequality and that each of the sampled trees yields an expectation value in the range of $[-1,1]$, we can conclude that with probability at least $1-\delta$,

\begin{equation}
        \Delta(f,\hat{\Tilde{f}})\leq  (|D|+|t|)^{\ell+1}+ \sqrt{\frac{2\log{(\delta^{-1}/2})}{K}}
\end{equation}
\end{proof}

\section{Comparison to other simulation methods}

\textbf{Comparison to other results:} Pauli backpropagation has proven to be a powerful technique for simulating quantum circuits, both in noiseless~\cite{rudolph2023,angrisani2024,lerch2024} and noisy~\cite{fontana2023,Shao2024,schuster2024,mele2024,gonzalezgarcía2024paulipathsimulationsnoisy} settings. These methods are motivated by theoretical guarantees that lay the foundation for their practical use, which are summarized in Table \ref{tab:table_pauli}. We see that our guarantees have different improvements over existing results in the literature: either they hold for larger classes of circuits, require less randomness or have better runtime guarantees.

For unital noise, such as depolarizing and dephasing noise,~\cite{fontana2023} and~\cite{Shao2024} have demonstrated that it is possible to approximate expectation values in average over rotation angles in polynomial time. These algorithms are both space and time efficient, enabling fast computations of expectation values. Similar results were reported in~\cite{schuster2024}, where the approximation holds on average over the input states. Additionally, assuming anti-concentration, accurate sampling from the quantum circuit was shown to be feasible.

However, extending these methods beyond unital noise remained challenging. For instance,~\cite{schuster2024} investigated only a randomized version of amplitude damping, while~\cite{Shao2024} had to include depolarizing noise on top of the amplitude damping channel to ensure their approximation held. For random circuits, a quasi-polynomial time Pauli backpropagation algorithm was derived (assuming geometrically local circuits) under non-unital noise \cite{mele2024}. However, under slightly stronger randomness assumptions, it has been shown that even noiseless circuits can be simulated using Pauli backpropagation techniques~\cite{angrisani2024}, with slightly worse run-time. In the noiseless regime, recent work has also established approximation bounds on certain regions of the energy landscape~\cite{lerch2024}. 

\begin{table}[h!]
\centering
\begin{tabular}
{|p{1.9cm}||p{2cm}|p{2.3cm}|p{1.9cm}|p{2cm}|p{2.2cm}|p{2.2cm}|p{2.2cm}|}
\hline
Reference & Ref. \cite{fontana2023,Shao2024} & Ref. \cite{schuster2024} & Ref. \cite{mele2024} & Ref. \cite{angrisani2024} & Ref. \cite{lerch2024} & Alg.~\ref{alg:Algorithm} & Alg.~\ref{alg:AlgorithmMC}\\ \hline
Noise model & Unital noise  & Depolarising + random amplitude damping & Non-unital noise & Noiseless & Noiseless & Non-unital noise & Non-unital noise\\ \hline
Circuit & Any circuit, random angles & Any circuit, random input states& Random circuits & Random circuits & Any circuit, random \textit{small} angles & Almost any circuit, random angles & Any circuit, random angles\\ \hline
Runtime $(\epsilon, p, n,m, D)$ & $n^2m(1/\epsilon)^{1/p}$ & $mn^{\log{(\sqrt{m+1}/\epsilon)}/p}$ & $e^{\log^D(\epsilon^{-1})}$ & $mn^{\log 1/\epsilon}$ & $m^{\log{1/\epsilon}}$ & $n^2m(1/\epsilon^2)^{3/2p}$ & $n^2m(1/\epsilon^2)^{1+1/p}$ \\ \hline
\end{tabular}
\caption{Comparison of the different Pauli backpropagation methods and their guarantees. The runtime may depend on the number of qubits $n$, the depth of the quantum circuit $m$, the circuit geometry $D$, the additive precision $\epsilon$ and the noise strength $p$.}
\label{tab:table_pauli}
\end{table}

\textbf{Future work:} A significant open question in the field concerns the extension of Pauli backpropagation methods to accommodate any noise model. In this work, we have made a great progress in this direction by handling a broad class of non-unital noise, but there still exist noise channels that fall beyond the scope of this work. Adapting to even more general noise models demands additional efforts and poses challenges due to the complexity of the backpropagated tree splitting. Furthermore, we have shown that for many noise models, it is possible to handle rotations along \emph{any} axis. It still remains to show whether this is true for all non-unital noise models. Another important problem is the simulation of noisy continuous time dynamics with noise, which could offer valuable insights into the performance of noisy quantum simulators. Further research in these areas could help understand the conditions under which quantum computations could offer real advantages.

\end{document}